\documentclass[sigconf]{acmart}
\renewcommand\footnotetextcopyrightpermission[1]{} 
\usepackage{amsmath,amssymb}
\usepackage{graphicx}
\graphicspath{ {images/} }

\usepackage{algorithm}
\usepackage{algpseudocode}
\usepackage{mdwmath}
\usepackage{mdwtab}
\usepackage[utf8]{inputenc}
\usepackage[english]{babel}
\usepackage{breqn}

\usepackage{pgfplots}
\usepackage{tikz}

\usepackage{bm}
\usepackage{mathtools}
\usepackage{amsthm}
\usepackage{xspace}
\usepackage{xcolor}
\usepackage[normalem]{ulem}
\usepackage{paralist}
\usepackage[normalem]{ulem}
\usepackage[T1]{fontenc}

\usepackage[font=small,skip=4pt]{caption}
\usepackage[font=small,skip=0pt]{subcaption}

\usepackage{color}
\usepackage{colortbl}
\usepackage{bbm}
\usepackage{url}

\usepackage[english]{babel} 

\usepackage[printwatermark]{xwatermark}

\usepackage{listings}
\usetikzlibrary{patterns}

\usetikzlibrary{matrix}
\usepgfplotslibrary{groupplots}
\pgfplotsset{compat=newest}

\newcommand{\advi}{{\mathcal A}}
\newcommand{\hsti}{{\mathcal H}}
\newcommand{\adv}{${\mathcal{A}}$\xspace}

\newcommand{\cset}{{\bf C}}
\newcommand{\st}{{\bf S}}
\newcommand{\pst}{{\bf P}}

\newcommand{\dex}{{DET}}

\newcommand{\hyb}{${\tt HYB}$}

\newcommand{\ths}{{\zeta}}
\newcommand{\intv}{{\mathbb I}}
\newcommand{\txp}{{\mathbb T}}
\newcommand{\cc}{{\mathbb C}}
\newcommand{\pc}{{\mathbb P}}
\newcommand{\queue}{{\mathbb Q}}

\newcommand{\trans}{{\bf P}}
\newcommand{\bcn}{{\bf B}}

\newcommand{\prot}{{\sc Tuxedo}}
\newcommand{\md}[1]{M/D/#1}

\newcommand{\gas}{{\sf gas}}
\newcommand{\type}{{\sf type}}
\newcommand{\typei}{{\sf type1}}
\newcommand{\typeii}{{\sf type2}}
\newcommand{\rcv}{{\sf to}}
\newcommand{\snd}{{\sf from}}

\newcommand{\toki}{{\sf token1}}
\newcommand{\tokii}{{\sf token2}}
\newcommand{\dep}{{\sf Deposit}}

\newcommand{\user}{${\sf user}$}
\newcommand{\contract}{${\sf contract}$}

\newcommand{\one}{\vspace{1mm}}

\newcommand{\vinay}[1]{\textcolor{purple}{{\bf Vinay:} #1}}

\newcommand{\rl}[1]{\textcolor{cyan}{{\bf Ling:} #1}}

\setcopyright{none} 

\begin{document}
\title{Better Late than Never; Scaling Computation in Blockchains by Delaying Execution} 

\author{Sourav Das}
\affiliation{%
  \institution{UIUC}}
\email{souravd2@illinois.edu}

\author{Nitin Awathare}
\affiliation{%
  \institution{IIT Bombay}}
\email{nitina@cse.iitb.ac.in}

\author{Ling Ren}
\affiliation{%
  \institution{UIUC}}
\email{renling@illinois.edu}

\author{Vinay J. Ribeiro}
\affiliation{%
  \institution{IIT Bombay}}
\email{vinayr@iitb.ac.in}

\author{Umesh Bellur}
\affiliation{%
  \institution{IIT Bombay}}
\email{umesh@cse.iitb.ac.in}


\begin{abstract}
Proof-of-Work~(PoW) based blockchains typically
allocate only a tiny fraction (e.g., less than 1\% for Ethereum) of
the average interarrival time~$(\intv)$ between blocks for validating
transactions.
A trivial increase in validation time~$(\tau)$
introduces the popularly known Verifier's Dilemma, and as 
we demonstrate, causes more forking and increases unfairness.
Large $\tau$ also reduces the tolerance for safety against
a Byzantine adversary.
Solutions that offload validation to a set of non-chain nodes 
(a.k.a. {\em off-chain}  approaches) suffer from trust
and performance issues 
that are non-trivial to resolve.

In this paper, we present \prot, the first on-chain 
protocol to theoretically scale $\tau/\intv \approx 1$ in
PoW blockchains. The key innovation in \prot\ is to separate
the consensus on the ordering of transactions from their 
execution. We achieve this by allowing miners to delay 
validation of transactions in a block by up to $\ths$ 
blocks, where $\ths$ is a system parameter. 
We perform security analysis of \prot\ considering all 
possible adversarial strategies in a synchronous network 
with maximum end-to-end delay $\Delta$ and demonstrate that \prot\ 
achieves security equivalent to known results for longest 
chain PoW Nakamoto consensus. Additionally, we also suggest 
a principled approach for practical choices of parameter 
$\ths$ as per the application requirement. Our prototype 
implementation of \prot\ atop Ethereum demonstrates that 
it can scale $\tau$ without suffering the harmful effects 
of na\"ive scaling in existing blockchains.
\end{abstract}

\maketitle
\pagestyle{plain}
\section{Introduction}
\label{sec:introduction}
One major problem of PoW blockchains such as Bitcoin and Ethereum 
is that they have useful compute power of orders of 
magnitude less than a typical desktop. For example, the gas limit
of each Ethereum block corresponds to a block processing time 
($\tau$) of approximately $150$ milliseconds.%
\footnote{Measured using a virtual machine with 16 cores, 
120GB memory, 6.4TB NVMe SSD, and 8.2 Gbps network bandwidth.} 
Thus only $1\%$ of the block interarrival time ($\intv$) of 15 seconds 
in Ethereum is used for executing transactions. This prevents 
permissionless blockchains from accepting blocks that contain 
computationally-heavy transactions.
Such computationally-heavy transactions are desirable for
applications such as cryptographic trusted setup and
privacy-preserving computation.

\noindent {\bf Problems due to large $\tau$.} 
To see why one cannot arbitrarily increase $\tau$, we must understand
the actions taken by a miner on receiving a new block.
The miner first validates 
the block by executing all its transactions. 
It then forms 
its own block whose transactions it executes, and finally starts
Proof-of-Work~(PoW) to mine a new block. 

Some permissionless systems, such as Ethereum, 
require each block to store a cryptographic digest of the latest blockchain state, which must be verified by miners while validating a block. This digest has many uses.
First, it assists users 
with low computation resources~(a.k.a., {\em light clients}) 
to efficiently validate, or prove to other light clients, a portion of the latest state. 
Second, it helps new nodes to quickly bootstrap and 
join the system. 
Third, it makes the system 
more compatible with stateless 
cryptocurrencies~\cite{chepurnoy2018edrax,tomescu2020aggregatable}. 

As a result, higher validation time $\tau$ ``eats into'' PoW time. 
This gives significant advantages to miners with higher block 
processing power than others, and also opens up the system to various 
attacks. As we demonstrate in~\textsection\ref{sec:evaluation} and
Appendix \ref{apx:high tau ethereuem}, with a larger $\tau$,
an adversary \adv who skips
validation of received blocks and/or its own created blocks 
can mine more than its fair share of blocks relative to 
its mining power on the main chain. For example, our experiments 
show that when $\tau/\intv=0.2$, an adversary \adv controlling 33\%
of the mining power can mine as much as 68\% of blocks
on the main chain. A large $\tau$ also leads to the 
well-known Verifier's Dilemma~\cite{luu2015demystifying} 
where a rational miner has to make the hard choice between 
validating blocks or not. The first choice
reduces its chances of mining the next 
block, and the second increasing its mining chances but
comes with the risk of accepting invalid blocks.
Moreover, increasing $\tau$ leads to 
a higher backlog of blocks to be processed at miners, 
which delays block forwarding. This leads to more forks 
and wasted mining power and lowers the adversarial tolerance 
of the system~\cite{pass2017analysis,ren2019analysis}.
For these reasons, PoW blockchains currently keep block 
validation time small relative to the block interarrival 
time, i.e. $\tau/\intv \ll 1$.

\noindent {\bf Previous Work.} 
Existing works that enable blocks with heavy 
computation do so through {\em off-chain} 
solutions~\cite{kalodner2018arbitrum,eberhardt2018zokrates,
bowe2018z,teutsch2017scalable,das2018yoda,cheng2019ekiden,
wust2019ace}. 
Rather than having all miners execute transactions, which 
we term {\em on-chain} computation, these methods 
delegate computation to a subset of miners or
groups of volunteer nodes. These solutions make additional 
security assumptions, beyond those required for PoW consensus, 
so that miners can validate 
the results that voluntary nodes submit.
Also, most off-chain solutions make restrictive assumptions 
about the interaction between contracts, e.g.,
one smart contract does not internally invoke functions 
of other smart contracts. 
Such interactions are desirable and often occur in 
practice (see Appendix~\ref{apx:interactive contracts}).
An {\em on-chain} solution if designed carefully can be 
made to automatically inherit the existing functionality of 
interaction between smart contracts and also the security 
guarantees of the underlying blockchain. 

\noindent {\bf Our Approach.} 
In this paper we propose \prot, the first on-chain 
solution that can theoretically scale $\tau/\intv$ close to $1$ 
while circumventing the problems discussed above.
As a result, \prot\ can increase the useful computing power
of PoW blockchains significantly to allow transactions with
non-trivial execution time. 

The core idea behind \prot\ is to separate the consensus on 
the ordering of transactions in the blockchain from their validation. 
We achieve this by allowing miners to delay validation of 
transactions in a block at height $i$ ($B_i$) until the arrival 
of the block at height $i+\ths$ ($B_{i+\ths}$) where 
$\ths$ is a system parameter. 
Essentially,  $B_{i+\ths}$ contains the cryptographic digest of the state corresponding to the execution of all transactions up to an including those in $B_i$. Hence we refer to this 
approach as {\em Delayed Execution of Transactions}~(DET). 
This way, the validation of transactions in a block can be 
done in parallel with the PoW, thereby side-stepping the 
competition between validation and PoW.

While the idea of DET may seem simple, securely adopting it
in PoW systems turns out to be non-trivial. The major challenge
arises due to the variability in the deadline (that is the arrival 
of $B_{i+\ths}$) for validating transactions (of $B_i$). As block
generation is a random process (often modeled as Poisson), there 
is the possibility (albeit rare) that an honest miner may fail 
to execute transactions in $B_i$ before receiving $B_{i+\ths}$. 
In that case, the honest miner will not be able to validate 
$B_{i+\ths}$ immediately on its arrival.  
As a remedy, in \prot\ an honest miner always extends its longest 
{\em validated} chain. 

It is also possible that a miner does not have the digest ready 
for the next block it wants to mine on the longest validated chain. 
For example, suppose $B_{i+\ths}$ is the last block in the longest validated chain and that the miner has executed all transactions in $B_i$. However, it has not yet validated $B_{i+1}$ and hence it does not have the digest ready to put into $B_{i+\ths+1}$ which it wants to mine.  
In such situations, \prot\ allows miners to put a special default 
state in place of the required state. Unless otherwise stated, we refer 
to this default state as the {\em empty state}. 

Using standard 
techniques from Queueing theory, we prove~(\S\ref{sub:reduction}) 
that these changes ensure that \prot\ achieves {\em Chain growth}, 
{\em Chain quality} and {\em Safety} similar to known results 
for Longest chain PoW~\cite{pass2017analysis,kiffer2018better,
ren2019analysis}. Also, we prove in~\S\ref{sub:zeta choice} that 
the fraction of blocks with empty state mined by the honest miner
can be reduced arbitrarily by setting $\ths$ appropriately. 

\noindent In summary, we make the following contributions:
\begin{itemize}
  \item We illustrate through analysis and experiments that a
  naive increase of $\tau$ in legacy blockchains gives unfair
  advantages to miners with faster processing power. 
  An adversary \adv can further exacerbate the unfairness 
  by skipping validation of received blocks and creating blocks
  that it can process quickly.

  \item We design \prot, a secure on-chain approach that can 
  theoretically scale $\tau/\intv$ to 1  in PoW based 
  permissionless blockchains.

  \item We theoretically prove security guarantees of \prot\ 
  under a synchronous network with end-to-end network 
  delay $\Delta$ and fixed processing time $\tau$.\footnote{Pass 
  et al.~\cite{pass2017analysis} use the term "asynchronous network" for a network with the same constraints.} Our analysis considers all
  possible strategies by a Byzantine adversary controlling up
  to $f_{\rm max} < 0.5$ fraction of the mining power. We also present
  an approach to choose $\ths$ in order to achieve any desired 
  fraction of honest blocks with non-empty state. 

  \item We implement \prot\ on top of the Ethereum Geth client and 
  evaluate it in an Oracle cloud with 50 virtual machines emulating the
  top 50 Ethereum miners. Our evaluation demonstrates that \prot\ 
  does not suffer from certain fairness problems unlike Ethereum does,
  for a high value of $\tau/\intv$.
\end{itemize}

\noindent{\bf Paper Organization.} 
In \S\ref{sec:system} we present our system model 
and assumptions. This is followed by a brief background of 
block validation process of legacy blockchains and attacks on 
them with high $\tau$ in \S\ref{sec:background}.
In \textsection\ref{sec:design} we introduce the concept 
of Delaying Execution of Transactions and describe how \prot\
employs it to achieve high validation time. We next describe 
our implementation methodology in \S\ref{sec:implementation}.
We then present our theoretical analysis demonstrating 
security of \prot\ in \textsection\ref{sec:analysis}. 
\textsection\ref{sec:evaluation} describes our prototype 
implementation of \prot\, experimental setup and observations from 
experimental results. We describe the related
work in \textsection\ref{sec:related work}. We discuss few FAQs in \textsection\ref{sec:faq} and conclude the work in \textsection\ref{sec:discussion}.

\section{System Model}
\label{sec:system}
We consider a permissionless system consisting of 
a set of miners. These miners form a connected 
network and run a blockchain protocol with Proof-of-Work~(PoW)
as the underlying consensus.
All honest miners mine blocks on top of the 
longest validated chain known to them 
(see~\textsection\ref{sub:dxt challenges}).
Block generation in \prot\ is assumed to follow a Poisson process 
with the rate $\lambda$ where $\lambda$ depends on the mining power 
of the network and difficulty of the PoW
puzzle. Each miner $n_a$ controls $p_a$ fraction of the mining power. 
Hence, any arbitrary miner $n_a$ will generate blocks at a rate 
$\lambda_a=p_a\lambda$. \prot\ allows execution of Turing Complete 
programs called {\em Smart Contracts}. A smart contract can be 
created by sending a transaction to deploy it on the blockchain. 
Once a  contract appears in the blockchain, its exposed functionality 
can be invoked by other miners through transactions. 

Smart contracts in \prot\ have unique IDs, and they maintain 
state, where state corresponds to the unique set of 
key-value pairs stored at each miner and is controlled by 
the program logic of the smart contract. For any arbitrary 
smart contract $c_z$, we use $\sigma_z$ to denote the state of 
smart contract. In addition to contracts, 
\prot\ maintains {\em Accounts}, which maintains tokens. 
Each account also has a globally unique ID which is the public key 
of a public-private key pair generated from a secure asymmetric 
signature scheme. Additionally, \prot\ has {\em Clients} which own accounts, can generate transactions to create smart
contracts, invoke their functions and transfer tokens from one 
account to another.

Transactions in \prot\ are ordered in a {\em Transaction Ordered 
List} (TOL) and are included in a block. We use $T_i$ to denote
the $i^{\rm th}$ TOL. The contracts generated by transactions 
in TOLs $\{T_1,T_2,\ldots,T_i\}$ are denoted as 
\cset$_i = \{c_z|z=1,2,\cdots\}$ and the corresponding state 
as \st$_i=\{\sigma_z|z=1,2,\cdots\}$. 
Each miner locally maintains states and updates it by executing a 
given TOL. Formally, with initial state $\st_{i-1}$, the execution 
of the TOL $T_i$ is denoted by:
\begin{equation}
\st_i = \Pi(\st_{i-1}, T_i)
\end{equation}
where $\Pi$ denotes the deterministic state transition function that 
executes transactions in $T_i$ in the order they appear. 

Let $\bcn_l=\{B_0,B_1,\cdots,B_l\}$ be the blocks known to $n_a$. 
In addition to state, $n_a$ maintains a transaction pool 
$\txp^{(a)}_l$, which contains the set of valid transactions created by 
clients that are yet to be included in a block till $B_l$. 
Hereon, when clear from the context, we drop the superscript from 
$\txp_l^{(a)}$ for ease of notation.

\vspace{0.5mm}
\noindent{\bf Assumptions.}
We assume the underlying network to be synchronous with end-to-end delay
of at most $\Delta$, i.e., all messages sent by an honest miner gets 
delivered to every other honest miner within time $\Delta$ from its 
release.
Also, we assume that all honest miners process blocks in any 
particular chain serially at the rate of $1/\tau$, where $\tau$ is 
the maximum time needed to validate a block. Like Ethereum, this 
can be enforced by requiring each transaction to specify the 
maximum time needed for its execution and keeping a cap $\tau$ 
on the total execution time a block. 
Note that block processing is different from mining a block; 
mining involves solving the PoW puzzle, whereas processing is 
about executing the transactions inside a block. Also, a block 
with high validation time neither implies a large block size 
nor that the block has a large number of transactions in it. 
A small block containing a few computationally-heavy transactions 
can require a large validation time.

We also assume that the block processing at an honest
miner does not contend with PoW, and the honest miners can 
simultaneously process blocks in distinct forks of the 
blockchain tree. We envision that our system will be adopted
by blockchains such as Bitcoin and Ethereum where block
processing can be done using CPUs while mining requires 
ASICs. Furthermore, the number of simultaneous forks in them
are quite small~\cite{gencer2018decentralization}.

We assume the presence of an adversary \adv, who can control 
up to $f_{\rm max} < 1/2$ fraction of total mining power of the network
and generate blocks at a rate $\beta=f_{\rm max}\lambda$.  
Adversarial miners can be Byzantine and can deviate 
arbitrarily from the specified protocol. 
The remaining miners are honest, control the remaining $(1-f_{\rm max})$
the fraction of the mining power, generates blocks at a rate 
$\alpha=(1-f_{\rm max})\lambda$,
and strictly follow the specified protocol. \adv can see every 
message sent by honest parties immediately and can inject its
messages at any point in time. Also, \adv can delay messages
sent by the honest parties by a maximum of $\Delta$ time.

\section{Block Validation in Legacy Blockchains}
\label{sec:background}
In this section, we first give some background on the block validation
process in Ethereum and later demonstrate why increasing $\tau$
leads to reduced fairness in terms of the fraction of blocks mined by an
honest miner. This background assists us in identifying the core 
problem behind smaller block validation time in existing systems. 
In the later sections~(\S\ref{sec:design}) we will describe how 
\prot\ securely addresses these issues. 

Ethereum is designed to force miners to validate blocks that 
they receive. To understand why, we must note that
Ethereum state is not explicitly stored on the blockchain, 
only its digest is. 
For a miner to create a potential block of its own which 
includes a mining reward, it must put a digest of the new state
resulting from this block in the header. However, the state
resulting from executing the transactions of the previous block 
acts as an essential starting point to obtaining the correct 
state to put in its own block. Hence the miner is forced to 
execute transactions in the previous block which it received.

Once a miner successfully creates its own block, it starts
mining, i.e., solving PoW on this block. During this PoW two
things can happen: {\em first}, the miner receives a conflicting
block created by a different miner at the same height as its 
own potential block; and {\em second}, the miner successfully 
solves the PoW, broadcasts its own block with the valid PoW and 
proceeds to create the next block, extending its very own recent 
block. Let us refer to the time spent validating the received 
block as the {\em Validation} phase; the time spent in creation 
of next block as the {\em Creation} phase, and time spent in 
PoW as the {\em Mining} phase. 
As honest miners do not solve for PoW during the validation and
creation phase, we are interested in the time these phases 
takes to complete. For this purpose, we will first scrutinize
these phases more carefully.



\begin{figure}[t!]
    \centering
    \includegraphics[width=0.90\linewidth]{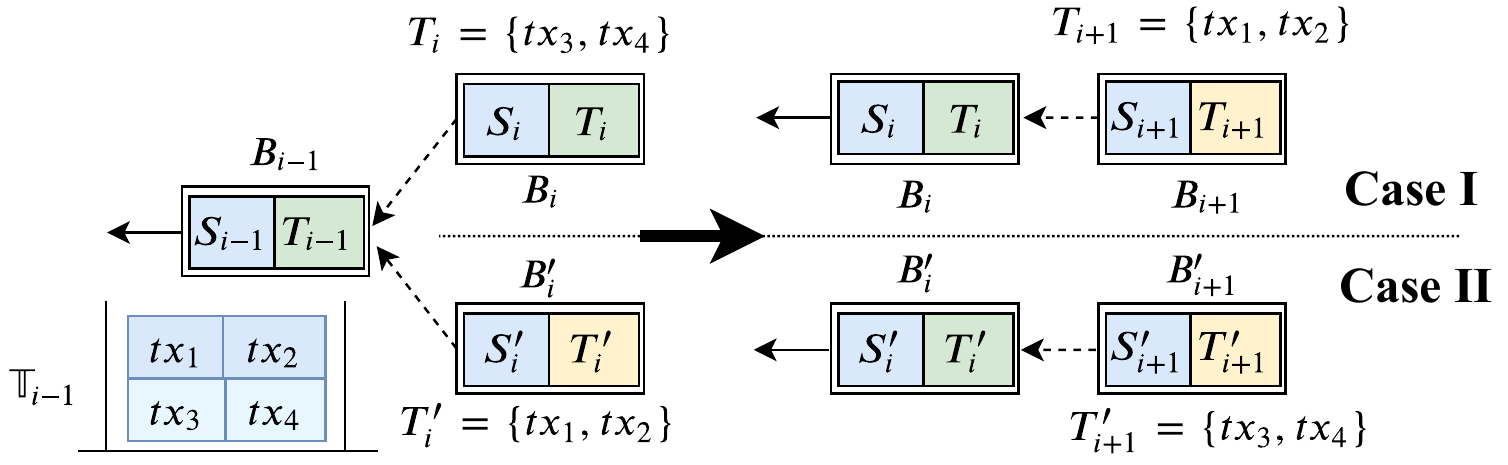}
    \caption{State transition at an Ethereum miner with $i-1$ 
    blocks on arrival of the $i^{th}$ block. Dashed arrows 
    represent next potential block and solid arrow represents
    actual arrival of a block. In case I, $n_a$ receives $B_i$,
    a block  mined by a different miner and in case II, $n_a$ mines 
    the block $B'_i$ by itself.}
\label{fig:eth block validation}
\end{figure}



If the longest chain known to an honest miner $n_a$ has length 
$i-1$, with $B_{i-1}$ at the tail of the chain as shown in 
Figure~\ref{fig:eth block validation}, $n_a$ tries to mine
the next block $B'_i$ at height $i$. Let $\st_{i-1}$ be
the state after executing TOLs till $B_{i-1}$ and 
$\txp_{i-1}=\{tx_1,tx_2,tx_3,tx_4\}$ be the latest transaction 
pool. For $B'_{i}$, $n_a$ first picks up a TOL $T'_{i}$ 
(e.g., $T'_{i}=\{tx_1, tx_2\}$), executes its transactions 
in order of their occurrence and starts PoW on block $B'_i$.
Let $\st'_{i}=\Pi(\st_{i-1},T'_{i})$ be the updated state. 
Note that till $n_a$ successfully solves the PoW puzzle, the 
updated state is not committed
and remains cached at $n_a$.
As described earlier, while running the PoW algorithm for 
block $B'_{i}$, one of two things can happen: 
either $n_a$ receives a valid block $B_{i}$ at height $i$ 
from the network or $n_a$ successfully solves the PoW. 
We now describe these as Case I and II, respectively.

\vspace{0.5mm}
\noindent{\bf Case I.} 
Let $n_b$ be the miner of the block $B_i$ (containing ordered 
list $T_i$) that $n_a$ receives. 
Without any coordination between $n_a$ and $n_b$, it is likely 
that $T_{i} \ne T'_{i}$. In that case, $n_a$ first validates 
$B_{i}$ executing all the transactions in $T_{i}$. 
On successful validation, $n_a$ accepts the block and proceeds 
to create the block $B_{i+1}$ at height $i+1$ by picking a new 
TOL $T_{i+1}$ from $\txp_{i} =\txp_{i-1}\setminus T_{i}$. Case I 
of Figure~\ref{fig:eth block validation} illustrates this.

\vspace{0.5mm}
\noindent{\bf Case II.} Unlike Case I, $n_a$ commits the state update
due to execution of $T'_{i}$ and proceeds to create the block 
$B'_{i+1}$ at height $i+1$ after picking a new TOL $T'_{i+1}$ from 
$\txp_{i-1}\setminus T'_{i}$. 
In our example, $T'_{i+1}=\{tx_3, tx_4\}$. $n_a$ then executes the 
new TOL $T'_{i+1}$ and starts PoW for $B'_{i+1}$. Case II of 
Figure~\ref{fig:eth block validation} illustrates this.
\begin{figure*}[!t]
    \centering
    \includegraphics[height=5.0cm, width=0.80\textwidth]{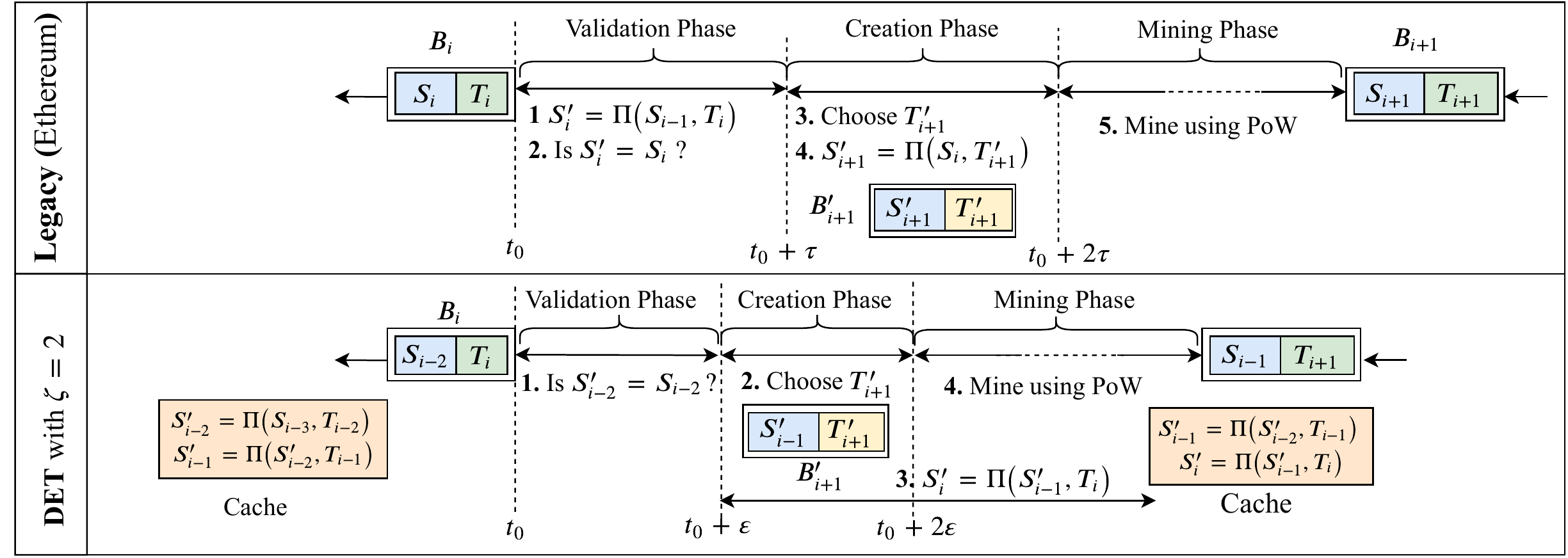}
    \caption{Actions taken by a miner on receiving a block $B_i$ at 
    time $t_0$ to validate $B_i$ and create the next block at height 
    $i+1$ in Ethereum~(top) and DET with $\ths=2$~(bottom). In DET, 
    $\epsilon$ denotes the time spent by the miner to validate the 
    received block and create the next one. Since validation/creation 
    in DET involves only a small (constant) number of operations, 
    $\epsilon \ll \tau$.}
    \label{fig:validation compare}
\end{figure*}

Another way to describe the block validation and creation 
mechanism is through the time intervals during which 
a miner validates the received block, creates the next block 
and performs PoW. 
Let $t_0$ be the time instant of the arrival of the block at 
height $i$ and let $\tau$ be the validation and creation time of 
a block. In case I, the miner validates the received block in
the time interval $(t_0, t_0+\tau)$, creates the next block 
$B'_{i+1}$ in time interval $(t_0+\tau, t_0+2\tau)$ and only at 
time $t+2\tau$ starts PoW for $B'_{i+1}$. 
However in case II, since $n_a$ himself is the creator of the 
block $B'_i$, it skips validation of $B'_i$ and spends the time 
$(t_0, t_0+\tau)$ in creating block $B'_{i+1}$, and starts PoW 
for $B'_{i+1}$ at time $t_0+\tau$. 
The top half of Figure~\ref{fig:validation compare} demonstrates 
the timings along with the computations a miner needs to perform 
for Case I in more detail.

\subsection{Consequences of high $\tau$ in legacy PoW Blockchains}
\label{sub:high tau ethereuem}
Ideally when all miners are honest and with no network delay, 
one would expect that the fraction of blocks mined by a miner should
be proportional to its mining power.
In this section, we demonstrate that is not the case, and show that 
when $\tau/\intv$ is high, the fraction of block mined by
honest miners heavily depends on their relative transaction 
processing speed in addition to their mining powers. 

Observe from Case II that the creator of a block
spends $\tau$ units of extra time (i.e., between $(t_0+\tau, t_0+2\tau)$)
for PoW while the remaining miners are busy creating the next block. 
This extra time $\tau$ increases its chances of mining the next block as well. 
This effect gets exacerbated if the miner controls a large mining 
power (say 30\%), because the miner will naturally mine blocks frequently 
and each of these blocks gives it an advantage to mine the next one as well.

More concretely, let $\lambda_a$ and $\tau_a$ be the block 
mining rate and block processing time of miner $n_a$, respectively. 
Let $c=\tau_a/\tau$ where $0\le c \le 1$, i.e $c$ is the ratio of 
block processing time of $n_a$ and remaining miners. $c=0$
implies that $n_a$ can process a block instantly independent of 
$\tau$. With these parameters, $n_a$ will spend only $2c\tau$
units of time in case I before starting PoW for the next block. 
Similarly, in case II, $n_a$ will spend only $c\tau$ units of
time creating the next block before starting PoW. Building on
this intuition, we theoretically compute the fraction of blocks
$n_a$ will mine in the longest chain for any given choice of
$\lambda_a$ and $c$ in Appendix~\ref{apx:high tau ethereuem}. 

Figure~\ref{fig:high tau honest} illustrates results from
our theoretical analysis. For example, with 
$\tau/\intv=0.26$, an honest miner who controls 30\% of the 
mining power and can validate or create blocks twice as fast as others,
i.e. $c=0.50$, will mine at least 46\% of the blocks. Further, a miner 
who skips both validation and creation of blocks, i.e with effective 
$c=0$ will mine at least 53\% of the blocks with 33\% of the 
mining power. We also measure the same using our experimental 
setup described in~\textsection\ref{sec:evaluation} with realistic network delays and observe 
that the network delay exacerbates the attack and allows \adv to mine 68\% of the 
blocks in the main chain.  
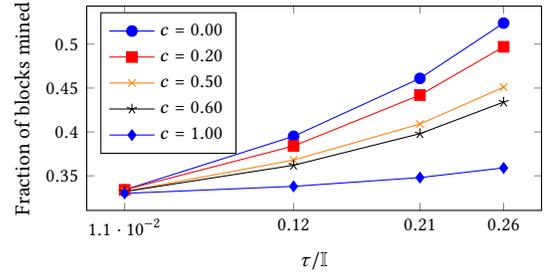
\begin{figure}[h!]
    \centering
    \pgfplotsset{footnotesize,height=4.3cm, width=0.90\linewidth}
    \begin{tikzpicture}
    \begin{axis}[
        legend pos=north west,
        ylabel=Fraction of blocks mined,
        xlabel=$\tau/\intv$,
        grid=minor,
        xtick={0.011, 0.122, 0.205, 0.260},
        ]
        \addplot table [x=t, y=c0, col sep=comma] {data/stat-honest.csv}; 
        \addplot table [x=t, y=c20, col sep=comma] {data/stat-honest.csv}; 
        \addplot [mark=x, orange] table [x=t, y=c50, col sep=comma] {data/stat-honest.csv}; 
        \addplot table [x=t, y=c60, col sep=comma] {data/stat-honest.csv}; 
        \addplot table [x=t, y=c100, col sep=comma] {data/stat-honest.csv};
        \addlegendentry{$c=0.00$}
        \addlegendentry{$c=0.20$}
        \addlegendentry{$c=0.50$}
        \addlegendentry{$c=0.60$}
        \addlegendentry{$c=1.00$}
    \end{axis}
    \end{tikzpicture}
    \caption{Fraction of blocks mined by miner $n_a$ with $~0.33$
    fraction of the total mining power, i.e $\lambda_a=\lambda/3$ for
    varying $c=\tau_a/\tau$.}
    \label{fig:high tau honest}
\end{figure}

\section{Design}
\label{sec:design}
In this section, we first give an overview of \prot. Then
we describe the concept of {\em Delayed Execution of 
Transactions}~(\dex), the core component of \prot,  
that enables us to make $\tau/\intv\approx1$. 
Next we describe a mechanism to make \dex\ robust against 
varying block interarrival time and a Byzantine adversary. 
Finally we describe the fee collection mechanism of \prot. 

\subsection{Overview}
\label{sub:overview}
Recall that a large transaction execution time ($\tau$) for a block is  unsuitable 
for a system like Ethereum. In Ethereum, while creating a block, a miner executes all its transactions and stores the resulting cryptographic digest in it. Likewise, miners receiving a block execute all its transactions to verify its digest. Since block validation, block creation, and PoW mining are sequential, a large creation or validation time eats into PoW mining time, opening up the system to unfairness and attacks.

%
In \prot, instead of making miners
execute transactions and reporting a cryptographic digest of 
the updated state immediately, we delay the reporting of the digest by $\ths \ge2$ blocks.
In particular the resulting state after executing all 
transactions up to a block at height $i$, $B_i$ is included in 
a block at height $i+\ths$. Intuitively, this allows miners 
to process received blocks and create new blocks in parallel 
to the PoW mining phase. Thus larger $\tau$ does not eat into PoW time.


A large $\tau$, however, introduces new scenarios not encountered
in PoW systems with negligible $\tau$, such as those studied hitherto~\cite{pass2017analysis,kiffer2018better,ren2019analysis}.
In one scenario, an honest miner may not be able to validate the digests present in its longest known chain. This happens if the longest chain has $B_{i+\ths}$ as its last block and the miner has not yet executed transactions in $B_i$, either because subsequent blocks were generated quickly by honest miners, or were generated privately by an adversary and released all at once. As a remedy, in \prot, a miner mines not on the longest known chain, but on the longest chain it has so far validated.
Intuitively, we must choose a 
large $\ths$ to reduce such occurrences. However, a very large $\ths$ is undesirable as 
it delays the reporting of the updated state. Hence, we must pick a suitable $\ths$ to balance this trade-off.



In another scenario, a miner may be able to validate all blocks in a chain, and yet not have the state ready to put in the next block it wants to mine. For example, if the last block is $B_{i+\ths}$ and the miner has executed transactions in $B_i$ but not $B_{i+1}$, then it can indeed verify the state in $B_{i+\ths}$ but does not have the state to put in $B_{i+\ths+1}$. \prot\ remedies this situation by allowing miners to report an
empty state (e.g. all zeros) in the block they create. 

We show in our analysis~\S\ref{sec:analysis}, that with the above remedies, the mining of honest miners never stalls, independent of 
any adversarial strategy. We exploit this fact to  lower bound the guarantees of \prot\ with known 
guarantees of PoW based Nakamoto consensus~\cite{pass2017analysis,kiffer2018better,ren2019analysis}.

\subsection{Delayed Execution of Transactions}
\label{sub:delayed execution}
\begin{figure}[t!]
    \centering
    \includegraphics[width=0.80\linewidth]{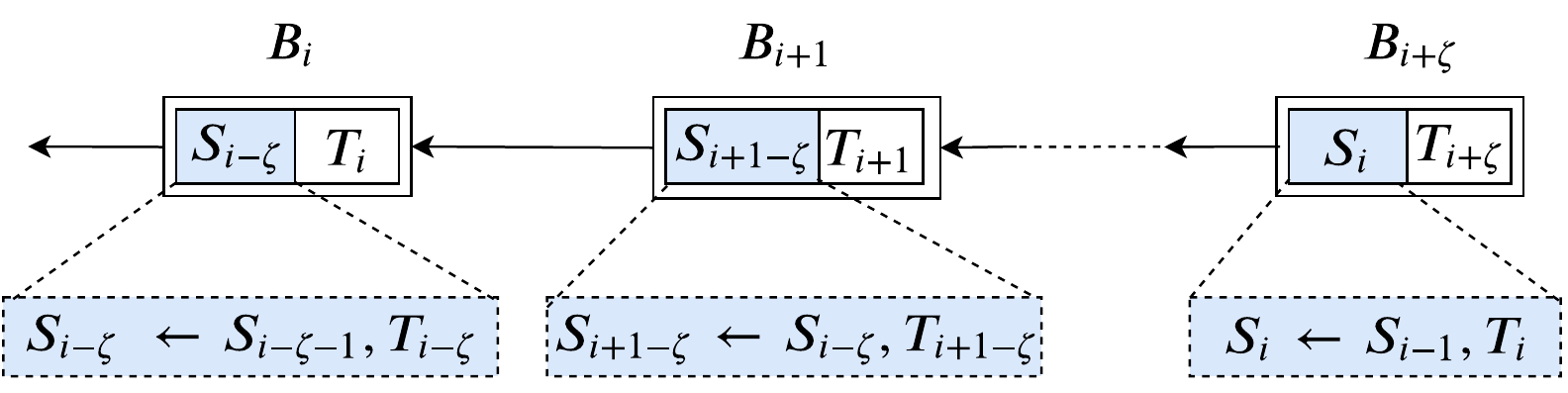}
    \caption{Structure of blocks with Delayed execution of 
    Transaction with $\ths$ blocks, where the state update 
    corresponding of transaction ordered list of a block 
    $B_{i}$ gets reported in the block $B_{i+\ths}$.}
    \label{fig:dxt blocks}
\end{figure}


The basic idea behind Delayed Execution of Transactions~(\dex) is 
to decouple the inclusion of 
transactions in the blockchain from the reporting (and hence 
validation) of the state resulting from those transactions. 
In \prot\, transactions are ordered in a block without being 
immediately validated, and the state resulting from them 
is reported $\ths$ blocks later. More formally, a block $B_i$ at 
height $i$ contains TOL $T_i$ and state 
$\st_{i-\ths}=\Pi(\st_{i-\ths-1},T_{i-\ths})$. Hence miners have 
a window of $\ths$ blocks to pre-compute the state required
for validation and this pre-computation can be done in parallel 
with the PoW. Figure~\ref{fig:validation compare} illustrates 
this for $\ths=2$ where the resulting states are delayed by 
2 blocks.

As explained in Case I in Section \ref{sec:background}, 
in existing blockchain designs,
up to $2\tau$ time can ``eat into'' the PoW time.
Thus, in order to get $\tau/\intv \approx 1$, 
we need to delay the execution of transactions by at least two blocks.
In the rest of this section, we first describe DET with $\ths=2$ and
then explain why an even larger $\ths$ is needed.

\vspace{0.5mm}
\noindent{\bf \dex\ with $\ths=2$.}
Let $B_{i-1}$ with state $\st_{i-3}$ and TOL $T_{i-1}$~(as $\ths=2$) 
be the latest valid block known to miner $n_a$~(see 
Figure~\ref{fig:validation compare}). For now, assume that $n_a$ 
has already computed 
(\romannumeral1) $\st'_{i-2}=\Pi(\st_{i-1},T_{i-2})$ and 
(\romannumeral2) $\st'_{i-1}=\Pi(\st_{i-2},T_{i-1})$ and cached 
them prior to the arrival of the block $B_{i}$. Here $\st'_{i-2}$
and $\st'_{i-1}$ are the states locally computed by $n_a$ for
TOL $T_{i-2}$ and $T_{i-1}$ respectively before arrival of block 
$B_{i}$~(ref. Figure~\ref{fig:validation compare}). Upon arrival 
of $B_{i}$, $n_a$ validates it by checking whether the reported 
$\st_{i-2}$ matches $\st'_{i-2}$~(step {\bf 1}). 
If it does, then $n_a$ accepts $B_i$ and starts computing 
$\Pi(\st'_{i-1},T_i)$~(step {\bf 3}). Simultaneously, 
$n_a$ picks a new TOL $T'_{i+1} \in \txp_{i-1}\setminus T_i$, 
creates the block $B'_{i+1}$ by fetching the precomputed 
state $\st'_{i-1}$ from its cache~(step {\bf 2}), and starts PoW 
for block $B'_{i+1}$~(step {\bf 4}).
This way, upon arrival of a block, $n_a$ is able to start PoW for 
the next block immediately.

\subsection{Handling variable block interarrival}
\label{sub:dxt challenges}
If blocks arrive exactly $\intv$ time apart from each other,
then $\ths=2$ will be sufficient to scale $\tau/\intv\approx1$. 
However, in reality, block interarrival times are random 
and can even be manipulated by the adversary to some extent. 
In case a sequence of $\ths$ blocks following
$B_i$ with TOL $T_i$ arrive closely spaced to each other, it is possible 
that a miner will not be able to compute the state $\st_i$ before
receiving $B_{i+\ths}$. Hence, the miner will not be able to 
immediately validate $B_{i+\ths}$. Without any precautionary measure,
in such a situation, miners will be forced to defer creation of the 
next block, and hence the PoW on it till it computes $\st_i$. 
If a large fraction of honest miners temporarily stop mining,
an adversary \adv with faster block processing power will effectively 
enjoy higher fraction of mining power and may even pull off the ``51\% attack'' during these periods. 
We address this issue by
making  two critical observations: 
{\em first,} the probability of this event occurring decreases with 
increasing $\ths$, and {\em second}, we can ask honest miners to mine
on the longest validated chain during such scenarios. 
We elaborate on these below.
%
%

\vspace{0.5mm}
\noindent{\bf Need for higher $\ths$.}
To see why increasing $\ths$ reduces the probability of the 
above mentioned undesirable event, we model DET as a queuing  
system where the transaction processing unit of a miner is 
analogous to the queue's server. Each arriving block is a task input to a queue  and each block is 
processed in $\tau$ units of time. In the 
absence of an adversary and network delays, the block arrival follows a 
Poisson process with rate $\alpha$.
Assuming the input rate is independent of the queue size, this is 
essentially an \md1 queue. The sequence of blocks that a miner is 
yet to process in a given chain represents the contents of the queue.

If an arriving block enters a queue of size less than or 
equal to $\ths-2$ then its own state as well as that of 
the subsequent block have been pre-computed. The probability 
of the queue exceeding $\ths-2$ is the probability of the 
miner missing the deadline for computing the state which that block must contain. 
This tail probability of the queue decreases with increasing 
$\ths$, thus making larger $\ths$ is more desirable. However,
there is a trade-off here, because a larger $\ths$ implies 
that blocks update the global state later, which is undesirable 
from a user's point of view. Hence $\ths$ must be chosen to 
balance this tradeoff. 
We leave detailed queuing models that take into consideration 
input variation of blocks due to $\Delta$, the presence of \adv, and the
trade-off due to larger $\ths$ to~\textsection\ref{sec:analysis}.



\noindent
{\bf Remark.} Due to forks, miners in \prot\ maintain multiple
queues, one for each forked chain (see Figure~\S\ref{fig:md1 queue}), 
and process them in parallel. Blocks which are common to multiple 
chain~(e.g. blocks $\{B_y,\cdots,B_i\}$ in Fig.~\ref{fig:md1 queue} 
need to be processed only once. 
All our analysis will be valid with multiple queues because, as assumed in \ref{sec:system},
a miner processes them in parallel and the input to each queue 
is still upper bounded by the block generation rate of miners.
We believe 
that the assumption of parallel processing of forks is a reasonable one because, in practice the 
number of simultaneous forks in them are quite
small~\cite{gencer2018decentralization} and is limited by the
block-generation capability of the adversary. Lastly, we envision 
that our system will be adopted by blockchains such as Bitcoin 
and Ethereum where block processing can be done using CPUs while
mining requires ASICs. Thus PoW mining and block processing do not compete for the same resources. 
\begin{figure}[!t]
    \centering
    \includegraphics[height=3cm, width=0.9\linewidth]{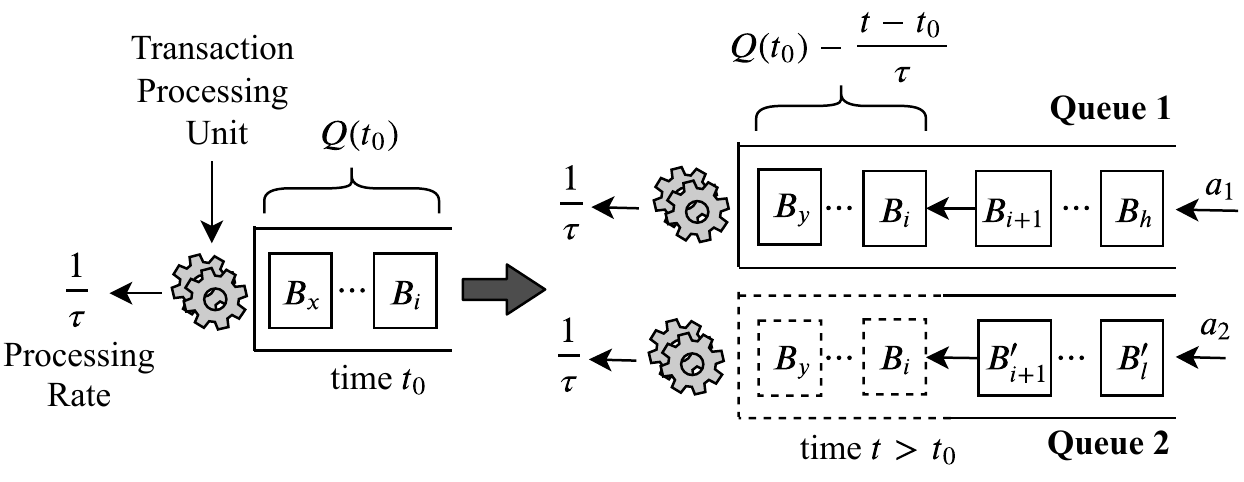}
    \caption{Queue(s) at a miner at time instant $t_0$ and $t>t_0$
    where the transaction processing unit at the miner process
    blocks in a queue at a rate $1/\tau$. In case of a fork,
    as shown on the right part of the Figure, i.e. $B'_{i+1}$ and 
    $B_{i+1}$ forking from $B_i$, miner maintains an additional
    queue for the new fork. The dashed part of the Queue 2 i.e.
    blocks $\{B_y\cdots B_i\}$ need not be processed explicitly 
    and results from processing these blocks from Queue 1 can 
    be directly used. Here $a_1$ and $a_2$ are arrival rates 
    to Queue 1 and Queue 2 respectively.}
    \label{fig:md1 queue}
\end{figure}

\noindent{\bf Extending Longest Validated Chain.}
Although higher $\ths$ lowers the probability of queue of an 
honest miner crossing $\ths-2$ in the absence of adversary,
additional care needs to be taken to provably prevent a 
Byzantine adversary from sabotaging the protocol. 
Thus we modify the chain selection rule of \prot\ from a standard 
longest chain selection procedure. Recall from~\ref{sub:overview}, 
the following changes are very crucial in lower bounding the chain
growth property of \prot\ which in turn is the core component
in the security analysis of any PoW based blockchain including 
\prot. Refer to~\textsection\ref{sec:analysis} for the detailed
security analysis of \prot.

Honest miners in \prot\ extend the {\em longest chain they can validate}. 
If a miner does not have the state to put in the next block, it puts a 
protocol specified default state, such as a sequence of zeros in place 
of the required state. Unless otherwise stated, we refer to this default 
state as the {\em empty state}. We refer to such blocks as {\em ES} 
blocks (i.e. blocks with an empty state). Similarly, we refer to blocks
with the non-empty state as {\em non-ES} blocks. ES blocks can contain
transactions (see~\textsection\ref{sec:implementation}) and non-ES block at
height $i$ report $\st_{i-\ths}$.  
Also, during the entire duration, honest miners continue to process
all unprocessed blocks with correct PoW that appears in a chain 
longer than the current mining head. On successfully  validating a block, a miner
re-configures its mining head to pick the new longest known validated
chain.

\subsection{Fees collection in \prot}
\label{sub:fee collection}
Every transaction in \prot\ specifies the maximum 
amount of computation resources needed for its execution. Based 
on this specification, the fee of every transaction in the $i^{\rm th}$
block, $B_i$ is collected in the same block. These fees are paid 
using the native token of \prot, $\toki$ (similar to Ether in 
Ethereum). Once the transaction gets executed, any leftover fees
i.e., fees of unused computational resources are refunded in 
$B_{i+\ths}$ where the state after the execution of $B_i$'s 
transactions is reported.\footnote{We leave the exact refund 
policy as a design choice as we primarily focus on the capability
of refunding fees if needed. \prot\ will work same for schemes 
that does not refund fees as well.} This ensures that only blocks
that can pay sufficient amount of gas, a unit of payment in 
Ethereum, for their transaction fees enter the blockchain. 

Note that additional care needs to set the minimum transaction fee paid by a transaction. In particular, we cannot 
levy a small fixed fee for every transaction as in Ethereum, 
as such a design can lead to under utilization in 
terms of actual amount of gas usage. In particular, malicious users may 
over-specify the amount of required computation resources and 
actually use only a tiny fraction of the specified resource. 
One way to discourage such behavior 
is to take fees 
equivalent to the minimum of a $\alpha$ ($0<\alpha\le 1$) 
times the specified gas usage and the true amount of 
gas used by the transaction. One many also consider 
alternative fee mechanisms depending upon the specific use 
cases. We leave the detailed analysis of the fee mechanism 
as  future research. 

Similar to Ethereum, \prot\ also allows its smart contracts 
to transfer and receive tokens. However, since the transactions of the
$i^{\rm th}$ block, $T_i$, are executed after fees are collected 
for $\ths-1$ future blocks, additional care needs to be taken to
prevent fees of future blocks from altering the execution results of
past transactions. Specifically, \prot\ restricts its smart 
contracts from using the native token. But, at the same 
time, \prot\ allows its contracts to create their own tokens 
reminiscent of ERC'20 tokens in Ethereum and use them during 
execution. These tokens could be contract-specific, shared
by several contracts, or shared by all contracts and is up 
to the contract designer. In our implementation, every smart
contract uses the single ERC'20 token which we refer to as
the $\tokii$. We describe the details of our implementation
in~\S\ref{sec:implementation}.

\section{Implementation Details}
\label{sec:implementation}
\begin{figure*}[!t]
    \centering
    \includegraphics[width=\textwidth]{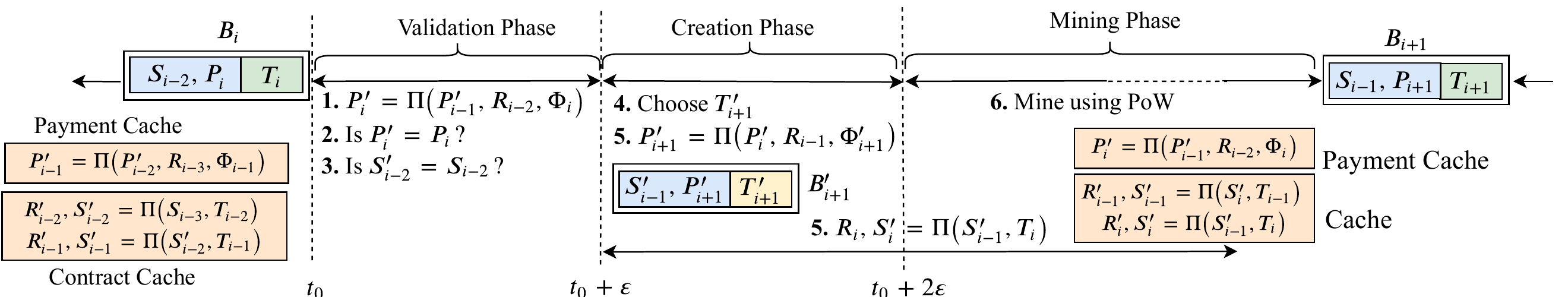}
    \caption{Steps performed at a miner to validate a block $B_i$ from 
    the network in \prot\ with $\ths=2$. After validation, miner 
    creates the next potential block $B'_{i+1}$ and starts PoW. During 
    PoW for block $B'_{i+1}$, the miner computes the state required to 
    instantly validate upcoming blocks.}
    \label{fig:tuxedo validation}
\end{figure*}
\noindent{\bf Accounts and Tokens.} Our implementation of \prot\ 
has two kinds of accounts: \user\ accounts and \contract\ accounts. 
Each \user\ account maintains both $\toki$, the native currency, 
and $\tokii$, the ERC'20 token to be used in contracts, whereas 
contract accounts only maintains $\tokii$. \contract\ accounts 
maintain executables that can be invoked by transactions.

\noindent{\bf Transactions.}
\label{para:transaction}
Each transaction $tx$ of \prot\ is a tuple containing 
$\{\type,\rcv,\snd,\gas,\star\}$ where $\type$ either takes the value
$1$ or $2$, $\rcv$ (resp. $\snd$) is the receiver (resp. sender)
address, $\gas$ specifies the maximum amount of gas $tx$ can use and
$\star$ represents the auxiliary information required for execution of
$tx$.  Transactions with $\typei$ are addressed to \user\ accounts and
transfers $\toki$ from the $\snd$ address to the $\rcv$ address. The amount of
$\toki$ transfer is present in the
auxiliary information denoted using $\star$. Transactions of $\typeii$
are addressed to \contract\ accounts and $\star$ contains the identity
of the functions to be invoked and the required function call parameters. For every such
transaction, miners in \prot\ transfer an amount of $\toki$, as a transaction fee from its
sender account to a prespecified address denoted by $\dep$. If in
case $tx.\snd$ does not have enough $\toki$, a miner discards $tx$.

\noindent{\bf Two States.}
\label{para:two states}
In our implementation of \prot\ every miner maintains two different
states: $\pst$ and $\st$ where $\pst$ is used to store information 
related to the amount of $\toki$ in each \user\ accounts and $\st$
is used to store the information regarding contract execution and 
amount of $\tokii$ in all the accounts. 
Since payment (including fees) and refund transactions only modify 
$\pst$, such a segregation enables faster validation and block
creation. 

\noindent{\bf Block Validation.}
Let $\st_{j-\ths}$ and $\pst_j$ be the contract and payment state at the
end of block $B_j$. Also, $R_{j-\ths}$ represents the refund 
processed after the execution of $T_{j-\ths}$. Let the latest block 
known to an honest miner $n_a$ be $B_{i-1}$. Let $B_i$ be the next
arriving block. Lets assume $B_i$ is a non-ES block and $n_a$ has 
already computed and cached the following state before its arrival.
\begin{align}
\pst'_{i-1} &=\Pi(\pst'_{i-2}, R_{i-\ths-1}, \Phi_{i-1}) \notag \\
R_{i-\ths}, \st'_{i-\ths} &=\Pi(\st'_{i-\ths-1}, T_{i-\ths}) \notag \\
R_{i-\ths+1}, \st'_{i-\ths+1} &=\Pi(\st'_{i-\ths}, T_{i-\ths+1}). \notag
\end{align}
Note that applying $\Pi$ on any TOL $T_j$ also 
outputs the ordered list of refund transactions corresponding to 
$T_j$. For $\ths=2$, this is depicted in Figure~\ref{fig:tuxedo validation}. 
On receiving the block $B_i$ containing TOL $T_i$ and digests of 
state $\pst_i$ and $\st_{i-\ths}$, $n_a$ validates $B_i$ as follows:
(\romannumeral1) $n_a$ first computes
$\pst'_i=\Pi(\pst'_{i-1},R_{i-\ths},\Phi_i)$,
(\romannumeral2) checks whether $\pst'_i$ matches with $\pst_i$, and
(\romannumeral3) $n_a$ also checks $\st'_{i-\ths}$ matches $\st_{i-\ths}$.

Alternatively, if $n_a$ has not pre-computed $\st'_{i-\ths}$ and $B_i$ is a 
non-ES block, $n_a$ continues to mine on the previous mining head till
it computes $\st'_{i-\ths}$ and re-starts validating $B_i$ as above. 
However, if $B_{i}$ is an ES-block and $n_a$ has already validated the latest
non-ES ancestor of $B_i$, $n_a$ computes $\pst'_{i}$ as 
$\Pi(\pst'_{i-1},\Phi_i)$ to check whether $\pst'_i$ matches with $\pst_i$ and skips step~(\romannumeral3) of validation. 
On successful validation, $n_a$ accepts $B_i$ and proceeds to create 
the next block as described below. Procedure {\sc Validate} in 
Algorithm~\ref{algo:tuxedo} presents the pseudo code for validation of 
blocks in \prot.

\noindent{\bf Block Creation.} 
\label{para:block creation}
On successful validation of the received block, to create the next
block $n_a$ 
(\romannumeral1) picks a new TOL $T'_{i+1}$ from $\txp \setminus T_i$,
(\romannumeral2) computes 
$\pst'_{i+1}=\Pi(\pst'_{i},R_{i-\ths+1},\Phi'_{i+1})$,
(\romannumeral2) fetches $\st'_{i-\ths+1}$ from cache (if available),
and (\romannumeral3) creates the next potential block $B'_{i+1}$ 
containing $T'_{i+1}$ and digests of $\pst'_{i+1}$ and $\st'_{i-\ths+1}$. 
Also, the first non-ES block after a sequence of ES-blocks applies
all accumulated refunds since the last non-ES block.
Alternatively, if $\st'_{i-\ths+1}$ is not available in the cache,
$n_a$ puts an {\em empty string} in place of $\st'_{i-\ths+1}$. 
After creating  $B'_{i+1}$, $n_a$ immediately starts PoW 
on $B'_{i+1}$. Procedure {\sc Create} in Algorithm~\ref{algo:tuxedo} 
presents the pseudo code for validation of blocks in \prot.

\noindent{\bf Execution of contract transactions}
\label{para:txn execution}
In \prot\, contract transactions are executed in parallel to PoW as
shown in Figure~\ref{fig:tuxedo validation}. Specifically, during PoW 
for $B'_{i+1}$, $n_a$ computes 
$R'_{i-\ths}, \st'_{i-\ths}=\Pi(\st'_{i-\ths-1}, T_{i-\ths})$. 
Also $n_a$ adds $T_{i}$ in the task 
queue~(ref.~\textsection\ref{sub:dxt challenges}) and executes $T_i$ 
as soon as it executes all $T_j$ for $j<i$, that appear prior to 
$T_{i}$.

\begin{algorithm}
    \small
    \caption{\prot\ for $\ths \ge 2$}\label{algo:tuxedo}
    \begin{algorithmic}[1]
        \State $\txp:\{tx_j|j=1,2,\cdots\}$ \Comment{Transaction pool at the miner}
        \State $\cc: \left\{R'_{j},\st'_{j}=\Pi(\st'_{j-1}, T_{j})\right\}$ \Comment{Contract Cache}
        \State $\pc: \left\{\pst'_{j}=\Pi(\pst'_{j-1}, R_{j-\ths}, \Phi_{j})\right\}$ \Comment{Payment Cache}
        \State $\queue:\{T_j|j=1,2,\cdots\}$ \Comment{TOL that are yet to be processed}
        \State \Call{ProcessTOL}{$\cdot$} \Comment{Non-blocking call to process existing TOL}
        \While{true}
            \State \Call{Reconfigure}{$B_k$} \Comment{On arrival of new block}
        \EndWhile
        \Procedure{Reconfigure}{$B_k$}              
            \If {\Call{Validate}{$B_k$}}
                \State {\bf stop} current PoW
                \State $B'_{k+1} \gets $ \Call{Create}{$B_k$}
                \State {\bf start} PoW on $B'_{k+1}$
            \EndIf
        \EndProcedure
        \Procedure{Validate}{$B_k$}
            \State $valid \gets$ {\bf false}; $\st_{k-\ths},\pst_k,T_k \gets B_k$
            \If {$\st_{k-\ths}$ is empty}
                \State  $\pst'_{k} \gets \Pi(\pst'_{k-1},\Phi_{k})$
                \If {$\pst'_k=\pst_k$}
                    \State $valid \gets$ {\bf true}
                \EndIf
            \Else
                \If {$\st'_{k-\ths}$ not in cache}
                    \State $valid \gets$ {\bf false}; add $T_k$ to $\queue$
                \Else
                    \State $\pst'_{k} \gets \Pi(\pst'_{k-1}, R'_{k-\ths},\Phi_{k})$
                    \If {$\pst'_k=\pst_k$ {\bf and} $\st'_{k-\ths}=\st_{k-\ths}$}
                        \State add $\pst_k$ to $\pc$; add $T_k$ to $\queue$; $\txp \gets \txp\setminus T_k$
                        \State $valid \gets$ {\bf true}
                    \EndIf
                \EndIf
            \EndIf            
            \State {\bf return} $valid$
        \EndProcedure
        \Procedure{Create}{$B_k$}
            \State $T'_{k+1} \gets$ subset of $\txp$
            \If {$\st'_{k+1-\ths}$ in cache}
                \State $\pst'_{k+1} \gets \Pi(\pst'_{k}, R'_{k+1-\ths},\Phi'_{k+1})$
                \State {\bf return} $(\st'_{k+1-\ths}, \pst'_{k+1}, T'_{k+1})$ 
            \Else
                \State $\pst'_{k+1} \gets \Pi(\pst'_{k},\Phi'_{k+1})$
                \State {\bf return} $(${\em empty-string}$, \pst'_{k+1}, T'_{k+1})$ 
            \EndIf
        \EndProcedure
        \Procedure{ProcessTOL}{$\cdot$}
            \While {true}
                \If {$\queue$ is non empty}
                    \State $B_{j} \gets$ next block in $\queue$
                    \State $T_j \gets$ TOL of $B_j$
                    \State $R_j,\st_{j} = \Pi(\st_{j-1},T_j)$
                    \State add $R_j,\st_{j}$ to $\cc$
                    \If {$j>$ current validated chain length}
                        \State Non-blocking \Call{Reconfigure}{$B_j$}
                    \EndIf
                \EndIf
            \EndWhile
        \EndProcedure
  \end{algorithmic}
\end{algorithm}

\section{Analysis}
\label{sec:analysis}
We analyze the security of \prot\ in the presence of a Byzantine
adversary under all possible adversarial strategies. Our analysis has the following outline. We start by modeling the blocks to be processed at an honest miner as a queuing system. 
Our queuing model captures the variability in the block inter-arrival 
time and the fact that honest miners might be processing old blocks 
as newer blocks continue to arrive. We then use our queuing 
analysis and the fact that honest miners are allowed to 
mine blocks with an empty state to illustrate that honest miners 
can extend blocks created by other honest miners even 
in a network with the worst possible latency. This implies that
the chain growth in \prot\ is greater than or equal to the chain
growth of PoW based Nakamoto consensus. This immediately 
implies that the lower bounds on chain-growth, chain quality 
and consistency shown
in~\cite{pass2017analysis,kiffer2018better,ren2019analysis} 
apply to \prot\ as well. Hence, \prot\ provides 
guarantees, equivalent to the known guarantees of existing 
PoW system.

\subsection{Block Processing as a Queuing System}
\label{sub:md1 model}
The arrival of blocks in PoW blockchain can be 
modeled as a Poisson process with arrival rate $\lambda$, or equivalently 
$1/\lambda$ is the expected inter-arrival time between 
two consecutive blocks~\cite{nakamoto2008bitcoin}. 
As all honest miners take $\tau$ units of time to process 
a block, i.e. the processing rate of the server is $1/\tau$. 
On arrival of every new block $B_i$ with TOL $T_i$ that
extends a chain longer than the current mining head at a
miner $n_a$, $n_a$ adds the block to its queue. 
$n_a$ processes (that is, validates) these blocks in First In 
First Out~(FIFO) order. As we have mentioned earlier, due
to forks, there will be multiple queues at each miner 
(see Figure~\ref{fig:md1 queue}), but our analysis applies
to any of them as arrival rate at each queue is dominated
by the arrival rate in a single queue setting and a miner processes
all queues in parallel.

Let $Q_a(t)$ denote the size of the queue of a miner $n_a$ at 
time $t$. If block $B_k$ enters the queue at time $t_k$, we use 
$Q_a(t_k^-)$ and $Q_a(t_k^+)$ to denote the size of queue 
immediately before and after time $t_k$ respectively. Note, 
$Q_a(t_k^+)=Q_a(t_k^-)+1$.

\vspace{0.5mm}
\noindent
{\bf Handling non-ES blocks.}
The ability miner $n_a$  to validate a received block $B_k$ and/or create
an non-ES block on it is directly related to the number of blocks in the
queue which $B_k$ enterss.
Notice that if $Q_a(t_k^+)>\ths$ then the head of the
queue contains TOL $T_i$ for $i\le k-\ths$, and the miner will not
be able immediately validate validate $B_k$. Similarly, when
$Q(t^+)=\ths$, the miner will be able to validate $B_k$ but
will not have the state to immediately mine a non-ES block on 
top of the received block.

\subsection{Reduction of \prot\ to Nakamoto PoW}
\label{sub:reduction}
In this section we illustrate that the {\em chain-growth} of \prot~(as defined in~\cite{pass2017analysis,kiffer2018better,ren2019analysis})
in a time interval $T$ is greater than or equal to the known chain
growth of PoW based Nakamoto consensus~\cite{pass2017analysis,kiffer2018better,ren2019analysis}.
As mentioned earlier, we later use this fact to prove security 
of \prot\ against all possible Byzantine adversaries.

Consider a time interval $[s,s+T]$. Select the honest blocks as
shown in Figure~\ref{fig:queue lemma}. First skip ahead to $s+\Delta$,
then find the next honest block, then skip by $\Delta$ and then
repeat till time $s+T-\Delta$. Let's call these blocks 
$B_1,B_2,\cdots,B_N$. Let $B_k$ be the honest block 
generated at time $t_k$.
\begin{figure}[t!]
\centering
\includegraphics[width=0.95\linewidth]{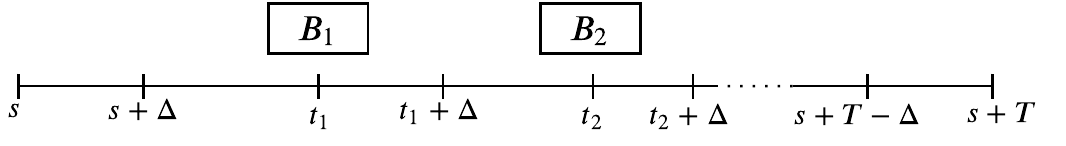}
\caption{Honest blocks chosen in time interval $[s+\Delta, s+T-\Delta]$
where the chosen blocks are separated by at least a gap of $\Delta$.}
\label{fig:queue lemma}
\end{figure}

\begin{lemma}
Let $B_k$ be the $m^{\rm th}$ block starting from genesis block 
in the chain containing $B_k$. Call these blocks 
$b_0, b_1,b_2,\cdots,b_m$ with $b_0$ as the genesis block.
Let $t'_0, t'_1,t'_2,\cdots,t'_m$ be the time when an honest 
miner hears the corresponding block for the first time.  
Then by $t_k+\Delta$ i.e. $t'_m+\Delta$, all honest miners 
would have processed the state required to validate $B_k$.
\label{lem:queue bound}
\end{lemma}
\begin{proof}
Let $Q_L$ and $Q_U$ be two hypothetical FIFO
queues with constant service rate $1/\tau$ in which blocks
$b_1,b_2,\cdots,b_m$ enter at $t'_1,t'_2,\cdots,t'_m$ and 
$t'_1+\Delta,t'_2+\Delta,\cdots, t'_m+\Delta$ respectively. 
Let $Q^{(k)}_a(t)$ be the position of the block $B_k$ (i.e. $b_m$) 
at the queue of miner $n_a$ at time $t$, and $Q_{B_k}(t)$ be its position at the miner which created it. Then the following two 
conditions hold:
\begin{align}
Q^{(k)}_L(t) &\le Q^{(k)}_a(t), \forall t \text{ after $B_k$ enters $Q_a$} \\
Q^{(k)}_a(t) &\le Q^{(k)}_U(t), \forall t \ge t_k+\Delta
\end{align}
Hence,
\begin{equation}
Q^{(k)}_a(t_k+\Delta) \le Q^{(k)}_U(t_k + \Delta) = Q^{(k)}_L(t_k) \le Q^{(k)}_{B_k}(t_k)
\label{eq:hst queue bound}
\end{equation}
Equation~(\ref{eq:hst queue bound}) implies that by time $t_k+\Delta$, the position of $B_k$ in the queue of
all honest miners will be less than or equal to the position of $B_k$ in the queue its creator at time $t_k$.
Hence if
 $Q_{B_k}(t_k^+) \le \ths$, i.e.  $B_k$ is a non-ES block, by
time $t_k+\Delta$, the position $B_k$ will be lower than or equal
to $\ths$ at all honest miners and hence, all honest miners will be
able to validate $B_k$ by time $t_k+\Delta$.
\end{proof}

Lemma~\ref{lem:queue bound} implies that whenever an honest miner
generates a block, $\Delta$ time  after the block 
generation time, 
every honest miner will have the state required
to validate the block generated by the honest miner. Hence,  
$\Delta$ time after generation of an honest block $B_k$, 
every 
honest miner extends a block which is at a height greater than or 
equal to the height of $B_k$. 
We use this argument in  Lemma~\ref{lem:height diff}
to prove that whenever two honest miners generate blocks at time 
instants that are at least $\Delta$ apart, the height of the 
latter block is greater than the height of the former. 
As a result, it is easy to see that blocks
$B_1,B_2,\cdots,B_N$ that we consider  have a strictly increasing height.

\begin{lemma}
Let $\ell(B_k)$ denote the length of the block starting from
genesis block $b_0$ with $\ell(b_0)=0$. Then $B_i$ for all 
$i\in [1,N]$ which were mined between $[s+\Delta, s+T-\Delta]$ 
as shown in Figure~\ref{fig:queue lemma} have distinct length. 
Further, $\ell(B_i) > \ell(B_j), \forall i>j$.
\label{lem:height diff}
\end{lemma}
\begin{proof}
Consider two consecutive blocks $B_k, B_{k+1}$ mined at time
$t_k, t_{k+1}$ respectively (need not be part of the same chain). 
Let $n_{k}, n_{k+1}$ be the miners of $B_k,B_{k+1}$ respectively. 
Since $t_{k+1} > t_k+\Delta$, $n_{k+1}$ would have heard of $B_{k}$ prior
mining $B_{k+1}$. Also, from Lemma~\ref{lem:queue bound} by time
$t_k+\Delta$, $n_{k+1}$ will have the state to validate $B_{k}$.
Thus from time $t_k+\Delta$ onwards, $n_{k+1}$ either will extend 
$B_{k}$ or any other validated block with same or greater length 
than $B_{k}$. This implies $\ell(B_{k+1})>\ell(B_{k})$. This is true
for all pair of consecutive blocks and hence by transitivity of
length comparison, we get $\ell(B_i) > \ell(B_j), \forall i>j$.
\end{proof}

Next we will use the Lemma~\ref{lem:queue bound} and 
Lemma~\ref{lem:height diff} show that during an interval 
of size $T$, the height of the blockchain at every honest 
miner grows by at least $N$.

\begin{lemma} {\em (Chain Increase)}
Let $L_j(t)$ be the length of the longest validated chain at miner
$n_j$ at time t. 
Let $L_{\rm min}(t), L_{\rm max}(t)$ be the minimum and maximum of 
chain lengths of all honest miners at any time $t$, i.e.
$L_{\rm min}(t) = \min_{j}\{L_j(t)\}$ and $L_{\rm max}(t)=\max_{j}\{L_j(t)\}$,
then in the scenario show in Figure~\ref{fig:queue lemma}, 
chain length of all honest miners grows by at least $N$ blocks, 
i.e $L_{\rm min}(s+T) \ge L_{\rm max}(s) + N$.
\label{lem:chain inc}
\end{lemma}
\begin{proof}
From Lemma~\ref{lem:height diff}, $\ell(B_{i+1})\ge \ell(B_{i})+1, \forall\ i\in[N-1]$.
Let $L_{B_{k}}(t_k^+)$ be the length of the longest chain of miner of block $B_{k}$ at
time $t_k^+$. Then from Lemma~\ref{lem:height diff} we know,
\begin{equation}
L_{B_{k+1}}(t_{k+1}^+) \ge L_{B_{k}}(t_k^+) + 1
\label{eq:length relation}
\end{equation}
From Lemma~\ref{lem:queue bound}, 
\begin{equation}
L_{\rm min}(t_k^++\Delta) \ge L_{B_{k}}(t_k^+) \notag
\end{equation}
Also, $L_{\rm min}(s+\Delta) \ge L_{\rm max}(s)$ as all blocks
generated before $s$ reaches every honest miner by 
time~$s+\Delta$. Hence,
\begin{align}
L_{\rm min}(s+T) &\ge L_{\rm min}(t_N+\Delta) \notag \\
& \ge L_{\rm min}(s+\Delta) + N \notag \ge L_{\rm max}(s) + N \notag \qedhere
\end{align}
\end{proof}

Notice that the scenario shown in Figure~\ref{fig:queue lemma}
is stochastically identical to the \hyb\ experiment shown by Pass et al. 
in~\cite{pass2017analysis}. Hence the chain growth and chain 
quality of \prot\ is identical to results presented 
in~\cite{pass2017analysis}. Also, it is easy to see that 
$B_1,B_2,\cdots,B_N$ includes all {\em loners} defined 
in~\cite{ren2019analysis} 
(or {\em convergence opportunities} defined in~\cite{pass2017analysis}) 
and hence the safety results
from~\cite{ren2019analysis} holds for \prot. For completeness
we state the theorem here.

\begin{theorem}
{\em (Safety (Thm. 8 in~\cite{ren2019analysis}))}
Let $B^*$ and $B^{**}$ be two distinct blocks at the same height. If 
$e^{-2\alpha\Delta}\alpha >(1+\delta)\beta$, then once an honest 
node adopts a chain that buries $B^*$ by $k$ blocks deep, no honest 
node will adopt a chain that buries $B^{**}$ by $k$ blocks, except 
for~$e^{-\Omega(\delta^2k)}$ probability. 
\label{thm:consistency}
\end{theorem}
\begin{proof}
Directly follows from Lemma~\ref{lem:chain inc} and proof of 
Theorem~8 of~\cite{ren2019analysis}.
\end{proof}


\subsection{Choice of $\ths$}
\label{sub:zeta choice}
As $Q(t)\ge \ths$ at a honest miner implies that the miner will not
be able to validate a received non-ES block, we compute an upper
bound on ${\rm Pr}[Q(t)\ge \ths]$ under all possible 
adversarial strategies after making certain approximations. Recall that $Q(t)\ge\ths$
do not violate security of \prot\ and hence the guarantees provided
in previous section still holds true.

A well-known result from queuing theory~\cite{norros1994storage} is that 
in any queuing system with constant service rate $1/\tau$, the size 
of the queue at any time $t$ is given by:
\begin{equation}
Q(t) = \sup_{s}\left\{A(s) - \frac{s}{\tau}\right\},
\label{eq:sup queue size}
\end{equation}
where $A(s)$ is the number of arrivals during the interval $[t-s,t]$.
In addition to the number of blocks generated during the time interval
$[t-s,t]$, $A(s)$ may also include honest blocks from time interval
$[t-s-\Delta, t-s]$ as these blocks might be delayed due to network. 
Furthermore, an adversary can deliberately withhold blocks mined prior
to time $t-s$ and release them during $[t-s,t]$. However, as demonstrated
in~\cite{pass2017analysis} that to withhold a block by longer than 
time $t_w$, adversary needs to generate a private chain longer than
honest chain during that time. 

Approximating the growth rate of the honest miners as a Poisson process
with rate $\gamma$ where $\gamma=\alpha/(1+\Delta\alpha)$, we can approximate 
the race between honest chain and the adversarial chain for time $t_w$ as a 
Skellam Distribution~\cite{skellam1946frequency} with $\mu_1=\beta t_w$
and $\mu_2=\gamma t_w$. Specifically, let $N(t_w), X_{\advi}(t_w)$ be the 
random variables denoting the chain growth of honest miner and number of
blocks generated by \adv during a time interval of size $t_w$ 
respectively. Then the success probability of \adv withholding a block
for longer than $t_w$ is ${\rm Pr}[X_{\advi}(t_w)-N(t_w) > 0]$. Since,
$X_{\advi}(t_w)$ and $N(t_w)$ are independent Poisson random variable,
$X_{\advi}(t_w)-N(t_w)$ follows a Skellam distribution with mean $\mu_1$
and $\mu_2$ as mentioned above.

Using results from Skellam distribution, given a small threshold $\eta$, we pick a value of $t^*$ such that
\begin{equation}
    {\rm Pr}[{X_\advi}(t^*)-Y(t^*)] \le \eta
    \label{eq:block withhold}
\end{equation}
and assume that \adv\ is not allowed to
withhold a block for more than $t^*$ units of time. Under this assumption,
we next upper bound the probability that queue of an honest miner will 
exceed any given $\ths$ under all possible adversarial strategies.

\begin{theorem}
For any given $\epsilon_0,\epsilon_1,t^*$, let
$s_0=\max\{\frac{\Delta}{\epsilon_0}, \frac{t^*}{\epsilon_1}\}$
and $\bar{\lambda} = (1+\epsilon_0)\alpha + (1+\epsilon_1)\beta.$ 
Let $Q(t)$ be the size of an honest miner's queue at time $t$. Then
\begin{equation}
{\rm Pr}[Q(t) \ge \ths] \le \sum_{i=\ths}^{\infty}\pi_{i} + 1-\sum_{i=0}^{\ths-1}\frac{\bar{\lambda}^ie^{-\bar{\lambda}s_0}}{i!},
\label{eq:zeta bound}
\end{equation}
where $\pi_{i}$ is the stationary distribution of \md1 queue
with arrival rate $\bar{\lambda}$.
\label{thm:queue size}
\end{theorem}
\begin{proof}
Let $X_{\hsti}(b), X_{\advi}(b)$ be the random variable denoting
the number of blocks mined by honest miners and adversary in 
a given time interval of length $b$ respectively. 
As we assume that \adv\ withholds a block for at most $t^*$ time
before the honest miner accepts them,
blocks in $A(s)$ are either mined by the adversary during 
$(t-s-t^*,t)$ or mined by honest nodes during $(t-s-\Delta, t)$. 
For any $\epsilon_0>0, \epsilon_1>0$, let 
$s_0=\max\{\frac{\Delta}{\epsilon_0}, \frac{t^*}{\epsilon_1}\}$. 
Then $\forall\ s\ge s_0, s+\Delta<(1+\epsilon_0)s$ 
and $s+t^* < (1+\epsilon_1)s$. Hence,
\begin{align}
A(s) &\le X_{\hsti}(s+\Delta) + X_{\advi}(s+t^*) \\
     &\le X_{\hsti}((1+\epsilon_0)s) + X_{\advi}((1+\epsilon_1)s), \forall\ s\ge s_0 
\end{align}

Let $X(s)$ be a random variable denoting the  number of blocks generated
by a Poisson process within a time interval of size $s$ with arrival rate 
$\bar{\lambda}=(1+\epsilon_0)\alpha + (1+\epsilon_1)\beta$.
Since independent Poisson random variables
are additive, we have the equality in distribution,
\begin{equation}
X(s) \stackrel{d}{=} X_{\hsti}((1+\epsilon_0)s) + X_{\advi}((1+\epsilon_1)s) 
\label{eq:joint poisson}
\end{equation}
\noindent
Hence using equation~\ref{eq:sup queue size}, we have  is
\begin{align}
    &\Pr[Q(t)>b] \\
    &\le {\rm Pr} \left[\bigcup_{s>0}\left\{A(s)-\frac{s}{\tau}>b\right\}\right] \tag{From equation~\ref{eq:sup queue size}} \\
    &= {\rm Pr}\left[\bigcup_{s>s_0}\left\{A(s)-\frac{s}{\tau}> b\right\}\right] 
    + {\rm Pr}\left[\bigcup_{s\le s_0}\left\{A(s)-\frac{s}{\tau} > b\right\}\right] \notag \\
    &\le {\rm Pr}\left[ \bigcup_{s>s_0}\left\{X(s)-\frac{s}{\tau} > b\right\}\right] + {\rm Pr}[A(s_0)>b] \\
    &\le {\rm Pr}\left[ \bigcup_{s>0}\left\{X(s)-\frac{s}{\tau} > b\right\}\right] + {\rm Pr}[A(s_0)>b]
    \label{eq:queue bound}
\end{align}
The first term of equation~\ref{eq:queue bound} is the standard
\md1 tail queue probability with arrival rate $\bar{\lambda}$, processing 
rate $1/\tau$ and hence its tail distribution probability decreases 
with increasing $\ths$. For any given $b$, $t^*$, $\epsilon$, and $\Delta$,
\begin{equation}
{\rm Pr}[A(s_0)>b] \le 1-\sum_{i=0}^{b-1}\frac{\bar{\lambda}^ie^{-\bar{\lambda} s_0}}{i!}.
\end{equation} \qed
\end{proof}

Using our worst-cast analysis, we suggest concrete values of $\ths$
one should consider to bound the probability of an honest miner's 
queue exceeding $\ths$. As we expect attacks to be intermittent 
(if any), we also numerically compute these bounds for an honest 
execution of the protocol, i.e., in a network without any 
adversary. Figure~\ref{fig:zeta est} and~\ref{fig:honest bound} 
plots the result of Theorem~\ref{thm:queue size} under 
some example parameters. 

\noindent
{\bf Concrete choice of $\ths$.}
For any given $\lambda, f_{\rm max}, \Delta,$ and $\tau$, we evaluate 
$\ths$ such that ${\rm Pr}[Q(t) \ge \ths]<0.01$. For all our evaluation 
we have used $\eta=0.001$ in equation~\ref{eq:block withhold}.
Figure~\ref{fig:zeta est} illustrates our results for different
values of $\tau/\intv$ and $\intv/\Delta$. For each $\tau/\intv$ and
$\intv/\Delta$, we pick $s_0$ that minimizes $\ths$. 
For example, with $25\%$ adversary and allowable processing time
equal to half of average interarrival time, i.e. $\tau/\intv=0.5$, 
we get $\ths=39$ for $\intv/\Delta=10$.
\pgfplotsset{small,label style={font=\fontsize{8}{9}\selectfont},legend style={font=\fontsize{7}{8}\selectfont},height=4.2cm,width=.90\linewidth}
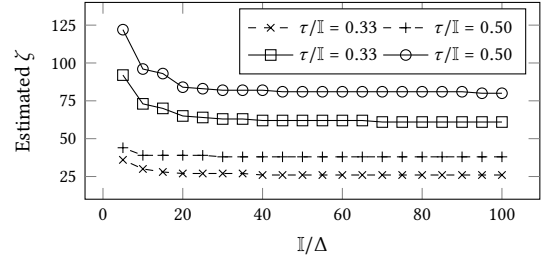
\begin{figure}[t!]
    \centering
    \begin{tikzpicture}
    \begin{axis}[
        legend pos=north east,
        legend columns=2,
        xlabel={$\intv/\Delta$},
        ylabel= {Estimated $\ths$},
        grid=none,
        ymin=10,
        ymax=140,
        ytick = {\empty},
        extra y ticks={25,50,75,100,125},
        extra y tick labels={25,50,75,100,125},
        ]
        \addplot[mark=x, dashed,
                mark options=solid, black] table [x=c, y=25-50, col sep=comma] {data/z-data.csv}; 
        \addplot[mark=+, dashed,
                mark options=solid,, black] table [x=c, y=25-75, col sep=comma] {data/z-data.csv};
        \addplot [mark=square, black] table [x=c, y=33-50, col sep=comma] {data/z-data.csv};
        \addplot [mark=o, black] table [x=c, y=33-75, col sep=comma] {data/z-data.csv};
        \addlegendentry{$\tau/\intv=0.33$}
        \addlegendentry{$\tau/\intv=0.50$}
        \addlegendentry{$\tau/\intv=0.33$}
        \addlegendentry{$\tau/\intv=0.50$}
    \end{axis}
    \end{tikzpicture}
    \caption{Required $\ths$ to upper bound Pr$[Q(t)\ge \ths]$ with $0.01$. Dashed lines correspond
    to $f_{\rm max}=0.25$ and solid lines are for $f_{\rm max}=0.33$.}
    \label{fig:zeta est}
\end{figure}

\noindent{\bf Remark.} 
It is important to note the event $Q(t)\ge \ths$ do not let the 
Byzantine adversary to violate the consistency of the protocol. 
Instead, it only allows the adversary to delay the reporting of 
update state for a very short duration of time. This is because,
when $Q(t)\ge \ths$ honest miners do not extend blocks mined 
by the adversary, and since the block processing rate is 
considerably higher than the block generation rate of the 
adversary, the queue at the honest miners will soon have less 
than $\ths$ blocks. Hence, very soon the honest miners will 
start creating blocks with non-empty state. 

\section{Evaluation}
\label{sec:evaluation}
\pgfplotsset{small,label style={font=\fontsize{8}{9}\selectfont},legend style={font=\fontsize{7}{8}\selectfont},height=3.7cm,width=1.1\textwidth}
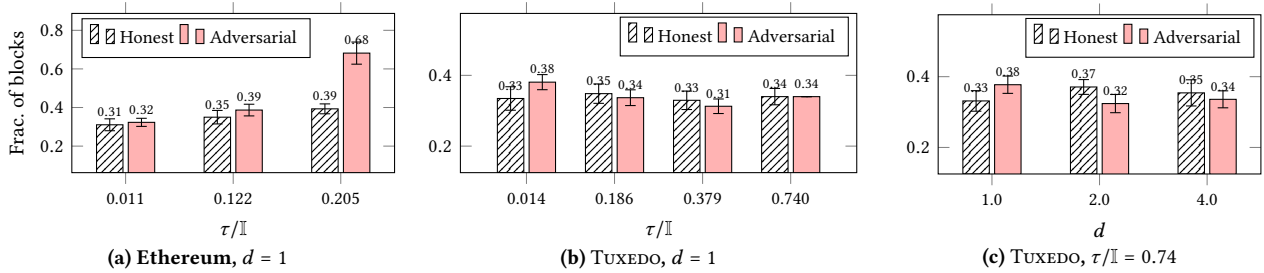
\begin{figure*}
\centering
\begin{subfigure}{0.30\linewidth}
\centering
\begin{tikzpicture}
\begin{axis}[
    ybar,
    enlargelimits=0.25,
    bar width=10pt,
    /pgfplots/ybar=2pt,
    legend columns=3,
    legend pos=north west,
    xlabel={$\tau/\intv$},
    ylabel= {Frac. of blocks},
    symbolic x coords={0.011, 0.122, 0.205},
    grid=minor,
    xtick={0.011, 0.122, 0.205},
    ymin=0.2,
    ymax=0.75,
    nodes near coords,
    every node near coord/.append style={font=\tiny},
    nodes near coords align={vertical},
    ]
    \addplot [pattern = north east lines] [error bars/.cd, y explicit,y dir=both,] table [x=avgHonestExTime, y=actFraction, y error=cfi,  col sep=comma] {data/ethNoSkip.csv}; 
    \addplot [fill=red!30!white] [error bars/.cd, y explicit,y dir=both,] table [x=avgHonestExTime, y=actFraction, y error=cfi, col sep=comma] {data/ethSkip.csv};
    \addlegendentry{Honest}
    \addlegendentry{Adversarial}
\end{axis}
\end{tikzpicture}
\caption{Ethereum, $d=1$}
\label{fig:eth high tau}
\end{subfigure}
\begin{subfigure}{0.35\linewidth}
\centering
\begin{tikzpicture}
\begin{axis}[
    ybar,
    enlargelimits=0.25,
    bar width=10pt,
    nodes near coords,
    /pgfplots/ybar=2pt,
    legend columns=2,
    xlabel={$\tau/\intv$},
    symbolic x coords={0.014, 0.186, 0.379, 0.740},
    grid=minor,
    xtick={0.014, 0.186, 0.379, 0.740},
    ymin=0.2,
    ymax=0.5,
    nodes near coords,
    every node near coord/.append style={font=\tiny},
    nodes near coords align={vertical},
    ]
    \addplot [pattern = north east lines] [error bars/.cd, y explicit,y dir=both,] table [x=exRatio, y=actFraction, y error=cfi, col sep=comma] {data/evdNoSkip.csv}; 
    \addplot [fill=red!30!white] [error bars/.cd, y explicit,y dir=both,] table [x=exRatio, y=actFraction,y error=cfi, col sep=comma] {data/evdSkip.csv};
    \addlegendentry{Honest}
    \addlegendentry{Adversarial}
\end{axis}
\end{tikzpicture}
\caption{\prot, $d=1$}
\label{fig:tuxedo high tau}
\end{subfigure}
\begin{subfigure}{0.30\linewidth}
\centering
\begin{tikzpicture}
\begin{axis}[
    ybar,
    enlargelimits=0.25,
    bar width=10pt,
    nodes near coords,
    /pgfplots/ybar=2pt,
    legend columns=2,
    xlabel={$d$},
    symbolic x coords={1.0,2.0,4.0},
    xtick={1.0,2.0,4.0},
    grid=minor,
    ymin=0.2,
    ymax=0.5,
    nodes near coords,
    every node near coord/.append style={font=\tiny},
    nodes near coords align={vertical},
    ]
    \addplot [pattern = north east lines] [error bars/.cd, y explicit,y dir=both,] table [x=delay, y=actFraction, y error=cfi, col sep=comma] {data/evdDelay.csv}; 
    \addplot [fill=red!30!white] [error bars/.cd, y explicit,y dir=both,] table [x=delay, y=actFraction,y error=cfi, col sep=comma] {data/evdDelaySkip.csv}; 
    \addlegendentry{Honest}
    \addlegendentry{Adversarial}
\end{axis}
\end{tikzpicture}
\caption{\prot, $\tau/\intv=0.74$}
\label{fig:tuxedo high delay}
\end{subfigure}
\caption{Fraction of blocks mined by the miner $n_1$ controlling
33\% mining power and $c=0.6$ in Ethereum and \prot}
\end{figure*}
\pgfplotsset{small,label style={font=\fontsize{8}{9}\selectfont},legend style={font=\fontsize{7}{8}\selectfont},height=3.7cm,width=1.1\textwidth}
\begin{figure*}
\centering
\begin{subfigure}{0.30\linewidth}
\centering
\begin{tikzpicture}
\begin{axis}[
    ybar,
    enlargelimits=0.25,
    bar width=10pt,
    /pgfplots/ybar=2pt,
    legend columns=3,
    legend pos=north east,
    xlabel={$\tau/\intv$},
    ylabel= {Mining Power Util.},
    symbolic x coords={0.011, 0.122, 0.205},
    grid=minor,
    xtick={0.011, 0.122, 0.205},
    ymin=55,
    ymax=115,
    nodes near coords,
    every node near coord/.append style={font=\tiny},
    nodes near coords align={vertical},
    ]
    \addplot [pattern = north east lines] [error bars/.cd, y explicit,y dir=both,] table [x=gasUsage, y=forkRate, y error=cfiFork,  col sep=comma] {data/eth-minefrac.csv}; 
    \addplot [fill=red!30!white] [error bars/.cd, y explicit,y dir=both,] table [x=gasUsage, y=forkRate, y error=cfiFork, col sep=comma] {data/eth-minefrac-skip.csv};
    \addlegendentry{Honest}
    \addlegendentry{Adversarial}
\end{axis}
\end{tikzpicture}
\caption{Ethereum, $d=1$}
\label{fig:eth mpu high tau}
\end{subfigure}
\begin{subfigure}{0.35\linewidth}
\centering
\begin{tikzpicture}
\begin{axis}[
    ybar,
    enlargelimits=0.25,
    bar width=10pt,
    nodes near coords,
    /pgfplots/ybar=2pt,
    legend columns=2,
    xlabel={$\tau/\intv$},
    symbolic x coords={0.014, 0.186, 0.379, 0.740},
    grid=minor,
    xtick={0.014, 0.186, 0.379, 0.740},
    ymin=55,
    ymax=115,
    nodes near coords,
    every node near coord/.append style={font=\tiny},
    nodes near coords align={vertical},
    ]
    \addplot [pattern = north east lines] [error bars/.cd, y explicit,y dir=both,] table [x=gasUsage, y=forkRate, y error=cfiFork, col sep=comma] {data/evd-minefrac.csv}; 
    \addplot [fill=red!30!white] [error bars/.cd, y explicit,y dir=both,] table [x=gasUsage, y=forkRate,y error=cfiFork, col sep=comma] {data/evd-minefrac-skip.csv};
    \addlegendentry{Honest}
    \addlegendentry{Adversarial}
\end{axis}
\end{tikzpicture}
\caption{\prot, $d=1$}
\label{fig:tuxedo mpu high tau}
\end{subfigure}
\begin{subfigure}{0.30\linewidth}
\centering
\begin{tikzpicture}
\begin{axis}[
    ybar,
    enlargelimits=0.25,
    bar width=10pt,
    nodes near coords,
    /pgfplots/ybar=2pt,
    legend columns=2,
    xlabel={$d$},
    symbolic x coords={1.0,2.0,4.0},
    xtick={1.0,2.0,4.0},
    grid=minor,
    ymin=55,
    ymax=115,
    nodes near coords,
    every node near coord/.append style={font=\tiny},
    nodes near coords align={vertical},
    ]
    \addplot [pattern = north east lines] [error bars/.cd, y explicit,y dir=both,] table [x=delay, y=forkRate, y error=cfiFork, col sep=comma] {data/evd-minefrac-delay.csv}; 
    \addplot [fill=red!30!white] [error bars/.cd, y explicit,y dir=both,] table [x=delay, y=forkRate,y error=cfiFork, col sep=comma] {data/evd-minefrac-delay-skip.csv}; 
    \addlegendentry{Honest}
    \addlegendentry{Adversarial}
\end{axis}
\end{tikzpicture}
\caption{\prot, $\tau/\intv=0.74$}
\label{fig:tuxedo mpu delay}
\end{subfigure}
\caption{Mining power utilization of the (a)Ethereum and (b),(c) \prot\ with increasing $\tau/\intv$ and Network delay.}
\end{figure*}
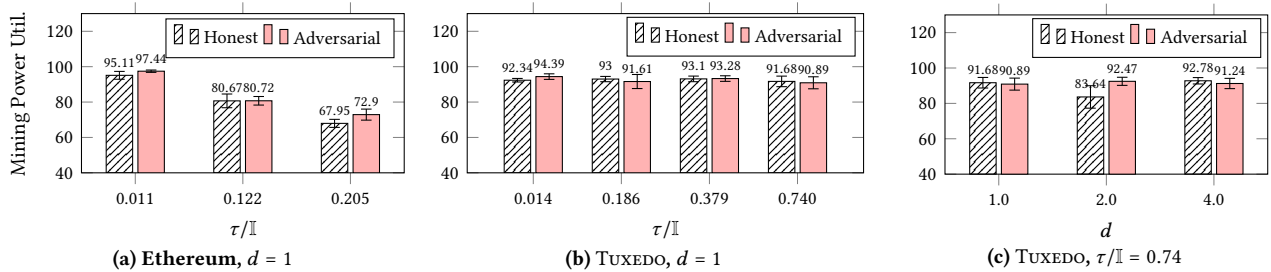

To evaluate computation scalability of \prot\ we built 
a prototype of on top of Ethereum Geth client 
version~${\tt 1.9.3}$. Our implementation consists of all
parts of \prot\ including an on-demand adversarial behavior
to skip validation of transactions. 
In many experiments we compare the performance of \prot\ with 
that of Ethereum. To facilitate such comparisons, we implement an 
adversary who skips validation and/or creation of blocks in
Ethereum.

\subsection{Experimental Setup}
\label{sub:expt setup}
Our experimental setup consists of 50 virtual machines (VMs) 
running in Oracle Cloud. All expect one of the VMs are dual-core 
machines with 8GM of RAM running Ubuntu.16.04. The remaining VM
which we use as an adversarial node has 8 core CPU @2.19 GHz,
30GB of RAM. We deliberately assign one node a computational
advantage over others to measure its effects on the fairness of Ethereum and
\prot. In our setup, all nodes have identical network 
bandwidth of 1GB/s download and 100MB/s upload speeds.

\noindent{\bf Node.}
Each VM in our experimental setup runs one \prot\ node. The mining 
power of each node is set according to the distribution of top 50 
Ethereum miners extracted from~\cite{ethMining}.
This corresponds to 99.98\% of Ethereum's total mining power. 
For each node in our setup, we simulate its block mining process 
by drawing the interarrival time between the blocks from an 
exponential distribution with parameter $\lambda/h$ where $h$
is the fraction of mining power controlled by the node. In 
Table~\ref{tab:mining fraction} reports the percentage of 
mining power controlled by top 14 miners that sum up to 97\% of
the total mining power.
\begin{table}[!h]
    \begin{center}
        \begin{tabular}{c c c c c c c}
        \hline
        32.98 & 16.16 & 15.06 & 5.72 & 5.67 & 4.41 & 4.14 \\
        3.53 & 2.61 & 1.84 & 1.34 & 1.32 & 1.25 & 1.05 \\
        \hline
        \end{tabular}
    \end{center}
    \label{fig:mining fraction}
    \caption{Percentage of mining power controlled by top 14
    miners in descending order our experimental setup. For 
    e.g. first and 14th miner controls $32.98\%$ and $1.05\%$
    of the mining power respectively.}
    \label{tab:mining fraction}
\end{table}

\noindent{\bf Network.}
To make a prototype of this Ethereum mining network, 
we collected data regarding the geographical location 
of the top 50 miners in Ethereum. Each node in our experiment 
emulates the geographical location of one such miner. We then
form a randomly connected network of these nodes where the degree
distribution follows a power law with exponent $-2.5$. 
Communication delays between every pair of nodes in the network 
are set accordingly to the ping delays observed between 
respective geographical locations~\cite{pingDelay}. We use Linux
${\tt tc}$ command to simulate the link delays.

\noindent{\bf Methodology.}
\label{sub:evaluation methodology}
We test \prot\ by deploying 50 contracts each implementing 
Quicksort, 2D matrix multiplication, and iteration with basic 
arithmetic operations. We then invoke functions from each
contract with appropriate parameters to achieve 
the desired block processing time. 
Throughout our experiment, we ensure that each block 
contains \textasciitilde165 transactions in total. 
As we simulate an adversary (\adv) that skips validation of 
blocks and creates new ones with contracts whose execution 
results are already known to the \adv, we deliberately restrict
all of the above mentioned contracts to be stateless. Note that,
as the primary metric of evaluating \prot\ is processing time of
a block, any choice of contracts will give us the same
results as long as they achieve the desired
block processing time.

\pgfplotsset{small,label style={font=\fontsize{8}{9}\selectfont},legend style={font=\fontsize{7}{8}\selectfont},height=3.9cm,width=1\textwidth}
\begin{figure*}[!htb]
\centering
\minipage{0.62\textwidth}
    \begin{subfigure}{0.50\linewidth}
    \centering
    \begin{tikzpicture}
    \begin{axis}[
        ybar,
        enlargelimits=0.25,
        bar width=8pt,
        /pgfplots/ybar=0pt,
        legend columns=2,
        legend pos=north east,
        xlabel={$\tau/\intv$},
        ylabel= {Mining Power Util.},
        symbolic x coords={0.011, 0.122, 0.205},
        grid=minor,
        xtick={0.011, 0.122, 0.205},
        ymin=55,
        ymax=125,
        ]
        \addplot [pattern = north east lines] [error bars/.cd, y explicit,y dir=both,] table [x=gasUsage, y=frac1, y error=cfi1,  col sep=comma] {data/eth-minefrac.csv}; 
        \addplot [fill=red!30!white] [error bars/.cd, y explicit,y dir=both,] table [x=gasUsage, y=frac1, y error=cfi1, col sep=comma] {data/eth-minefrac-skip.csv};
        \addplot [pattern = north west lines] [error bars/.cd, y explicit,y dir=both,] table [x=gasUsage, y=frac2, y error=cfi2,  col sep=comma] {data/eth-minefrac.csv}; 
        \addplot [fill=blue!30!white] [error bars/.cd, y explicit,y dir=both,] table [x=gasUsage, y=frac2, y error=cfi2, col sep=comma] {data/eth-minefrac-skip.csv};
        \addlegendentry{$n_1$,Hst}
        \addlegendentry{$n_1$,Adv}
        \addlegendentry{$n_2$,Hst}
        \addlegendentry{$n_2$,Adv}
    \end{axis}
    \end{tikzpicture}
    \caption{Ethereum, $d=1$}
    \label{fig:eth mpu miner high tau}
    \end{subfigure}
    \begin{subfigure}{0.50\linewidth}
    \centering
    \begin{tikzpicture}
    \begin{axis}
        [
        legend pos=north west,
        legend columns=1,
        symbolic x coords={0.011, 0.122, 0.205},
        xtick={0.011, 0.122, 0.205},
        xlabel= {$\tau/\intv$},
        ylabel= {Propagation delay},
        grid=minor,
        ]
        \addplot [red, mark=+, mark options=solid, error bars/.cd, y explicit, y dir=both,] 
        table [x=gas, y=noSkipMean, y error=noSkipCfi, col sep=comma] {data/eth-forkrate.csv};
        \addplot [blue, mark=x, error bars/.cd, y explicit,y dir=both,] 
        table [x=gas, y=skipMean, y error=skipCfi, col sep=comma] {data/eth-forkrate.csv}; 
        \addlegendentry{Honest}
        \addlegendentry{Adversarial}
    \end{axis}
    \end{tikzpicture}
    \caption{Ethereum, Median block prop. delay}
    \label{fig:eth high prop delay}
    \end{subfigure}
    \caption{(a) Mining power utilization of first two miner 
    ($n_1$ and $n_2$) (b) Measured median propagation delay of blocks with increasing $\tau/\intv$ in Ethereum.}
\endminipage\hfill
\minipage{0.35\textwidth}
    \centering
    \pgfplotsset{label style={font=\fontsize{8}{9}\selectfont},height=3.9cm,legend style={font=\fontsize{7}{8}\selectfont},width=\linewidth}
    \begin{tikzpicture}
    \begin{axis}[
        legend columns=1,
        xlabel={Queue length},
        ylabel= {Fraction of blocks},
        ] 
        \addplot [black, mark=x, error bars/.cd, y explicit,y dir=both,] table [x=queueLen, y=fracBlock240, y error=cfi240, col sep=comma] {data/queueEvdNoSkip.csv};
        \addplot [red, mark=+,error bars/.cd, y explicit,y dir=both,]table [x=queueLen, y=fracBlock240, y error=cfi240, col sep=comma] {data/queueEvdSkip.csv};
        \addlegendentry{$\tau/\intv=0.379$, hst.}
        \addlegendentry{$\tau/\intv=0.379$, adv.}
    \end{axis}
    \end{tikzpicture}
    \caption{Average queue size of honest miners (excluding $n_1$) for
    $\tau/\intv=0.379$ and $d=1$}
  \label{fig:tuxedo queue}
\endminipage\hfill
\end{figure*}

\subsection{Experiments and Results}
\label{sub:experimental results}
We first evaluate the effect of increasing 
$\tau/\intv$ in Ethereum with all miners being 
honest. We then repeat the experiment in the 
presence of an adversary \adv, who skips validation 
of received blocks and creates blocks with transactions for 
which \adv already knows the execution results. 
We then perform the same set of 
experiments with \prot\ and compare our
results with Ethereum.
Next for fixed $\tau/\intv=0.70$ we evaluate
\prot\ with increasing network delay. 

In all experiments the first miner ($n_1$) controls 
\textasciitilde33\% of the network's mining power and 
can process \textasciitilde1.67 times faster than other 
miners. In all experiments we keep block mining difficulty 
such that $1/\lambda=15.0$.
Hence $\intv$ is equal $2\tau+1/\lambda$ for Ethereum 
experiments and $1/\lambda$ for \prot\ experiments.

\noindent{\bf Fairness violations in Ethereum.}
Figure~\ref{fig:eth high tau} illustrates the fraction of
Ethereum blocks miner $n_1$ mines. Observe that with 
increasing $\tau/\intv$, the fraction of blocks $n_1$ mines 
increases even when $n_1$ is honest corroborating our theoretical 
analysis~(ref. Appendix~\ref{apx:high tau ethereuem}). 
When $n_1$ is adversarial, i.e. $n_1$ skips both validation
and creation of blocks, $n_1$ mines significantly higher
fraction of blocks. These fractions are higher than
theoretically computed fraction in Figure~\ref{fig:eth high tau}.
Because unlike our simplistic assumption in theoretical 
analysis~(ref.\ref{apx:high tau ethereuem}), in the 
experiments $n_1$ mines for the entire duration of the 
experiment.

\noindent{\bf High fork rate in Ethereum.}
With high $\tau/\intv$ in Ethereum, we observe that
the the fork rate of Ethereum increases. Let 
{\em Mining Power Utilization}~(MPU) of a blockchain network
be defined as --- the fraction of blocks mined by the
miners that end up in the eventual longest chain. 
A blockchain with lower MPU implies that a lower
fraction of mined blocks end up being in the blockchain and that many blocks are orphaned.
Figure~\ref{fig:eth mpu high tau} illustrates the MPU 
of Ethereum network with increasing $\tau/\intv$. 
Notice that despite having $\lambda=1/15$ 
and $\intv=2\tau+1/\lambda$, the fork rate increases.
This is because miners in Ethereum only forward a block
when they have fully validated its parent block. Thus
with high $\tau/\intv$, miners in Ethereum more frequently
encounter blocks whose parents are yet to be validated
by the miners. As a result, with $\tau/\intv$, the 
effective propagation delay of blocks in the network increase which
leads to higher forks and lowers the MPU of the network.

Figure~\ref{fig:eth high prop delay} illustrates that the 
median propagation delay of Ethereum network. Median 
block propagation delay in the presence
of a adversarial node is lower due to fact that the 
adversarial node can forward blocks immediately as it 
processes all received blocks immediately. 
Also we observe that higher $\tau/\intv$ does not lower
MPU of the adversarial node $n_1$ by much. Further, MPU
of a adversarial node is much higher than its honest 
counterpart. Figure~\ref{fig:eth mpu miner high tau} illustrates 
the mining power utilization of top 2 miner in our 
experimental setup~(ref.~\ref{tab:mining fraction}). 
This is because when $n_1$ is adversarial, it mines
solo for longer duration. This allows $n_1$ to mine
longer sequence of blocks more frequently while other 
miners were busy extending an older block. Hence, 
when network gets synchronized again, that is honest nodes have to backlog of blocks to process, the string of blocks mined by
the $n_1$ enters the blockchain with high 
probability.

\noindent{\bf Increasing $\tau/\intv$ in \prot.}
We demonstrate that higher $\tau/\intv$ does not affect
miners in \prot\ by repeating the above experiments in \prot.
Figure~\ref{fig:tuxedo high tau} illustrates that the fraction 
of blocks $n_1$ mines does not vary in \prot\ despite  high 
$\tau/\intv$ and an adversarial $n_1$. 
Also, unlike Ethereum high $\tau/\intv$ does not affect the
mining power utilization of \prot\ since all miners can immediately
validate all received blocks. This is illustrated in 
Figure~\ref{fig:tuxedo mpu high tau}. We use $\ths$ based
on our analysis in~\textsection\ref{sub:zeta choice} for
$t^*=0$ and ${\rm Pr}[Q(t)\ge\zeta] \le 2^{-15}$.

Figure~\ref{fig:tuxedo queue} demonstrates the queue observed by a 
arriving block in an honest miner for $\tau/\intv=0.379$. 
Observe that presence of a skipping adversary does alter the 
queue size honest miners observe. Further, more than $99\%$ of 
the blocks find a queue size of four or less and less 
than $<0.1\%$  blocks finds a queue size of eight or higher.
This implies that though $\ths$ is chosen to be high,
for majority of the blocks the users of \prot\ will get the
execution results of their transactions within four blocks. 

\noindent{\bf Increasing network delay.}
For fixed $\tau/\intv=0.74$ we evaluate \prot\ with 
increasing delay. Specifically we increase link delay
between each pair of connected node by a factor of 
$d=1,2,4$. Figure~\ref{fig:tuxedo high delay} illustrates
that higher delay does not affect the fraction of blocks
mined by $n_1$ for both honest and adversarial $n_1$. 
We also observe that the MPU of \prot\ does not decrease
in \prot\ despite higher network delay~(see Figure~\ref{fig:tuxedo mpu delay}). 
This is because,
the top five miners controlling approximately 75\% of mining
power in our experimental network are in close proximity 
with each other. Hence the increased delay does not affect
the block propagation delay between them. 

\section{Related Work}
\label{sec:related work}
To best of our knowledge, our work is the first on-chain solution to increase the ratio between block validation time and average 
block inter-arrival time in a PoW based blockchain.

There have been attempts to enable the execution of computationally
intensive smart contracts~\cite{kalodner2018arbitrum,
eberhardt2018zokrates,teutsch2017scalable,das2018yoda,cheng2019ekiden} 
through off-chain solutions where the execution of intensive transactions is delegated to a subset of 
miners or volunteer nodes. These solutions typically have high 
latency for off-chain computations and also make additional 
security assumptions beyond those required for PoW consensus.
Also, off-chain solutions restrict certain interactions between contracts, e.g., one smart 
contract cannot internally invoke functions of other smart
contracts. Such interactions are desirable and often occur 
in practice (see Appendix~\ref{apx:interactive contracts}).
Arbitrum~\cite{kalodner2018arbitrum} requires nodes participating
in the protocol to be rational and one participating node to be 
honest. Yoda~\cite{das2018yoda} requires an unbiased source of 
randomness whereas current mechanisms of generating distributed 
randomness are highly expensive~\cite{syta2017scalable}. 
Ekiden~\cite{cheng2019ekiden} relies on SGX Enclaves and requires 
all enclaves to be trusted, an assumption that is made questionable
by recent attacks~\cite{van2018foreshadow,brasser2017software}.
In Zokrates~\cite{eberhardt2018zokrates} participants are required 
to generate expensive non-interactive proofs for verification of 
off-chain computations.

A concurrent work ACE~\cite{wust2019ace} provides an off-chain 
solution that securely scales the smart contract execution and 
enables interactive calls between smart-contracts by 
implementing a variant of the classic two-phase commit 
protocol. Thus, interactive contracts in ACE have a very 
long delay as they can be delayed by the slowest committee 
involved in the interaction. Also ACE requires committee 
members to broadcast the updated state to the entire 
network. This can lead to large bandwidth usage for contracts 
that update a large number of variables. 
Furthermore, in ACE more than half of the committee members 
for each contract must be honest for guaranteed safety and 
liveness of the contract execution, and it relies on a 
reputation based mechanism to realize such assumptions. 
Note that this is a much stronger assumption than assuming, as \prot\ does, that
half of the total mining power is honest.

Note that \prot\ does not eliminate the benefits of off-chain 
based approaches. Instead, it further opens up the possibilities 
for more comprehensive (computationally intensive, if needed) 
on-chain aggregation schemes for off-chain schemes such as verifiable computation~\cite{lee2020replicated}. 
This can further scale the overall computation capacity of the entire system.

In an alternative approach, a recent line of work has tried to increase
the scalability of smart contracts by concurrently executing 
transactions~\cite{DickersonGHK17,anjanConcurrent,zhang2018enabling}. 
Dickerson et al.~\cite{DickersonGHK17} enables miner to concurrently 
execute transactions using a pessimistic abstract lock and inverse-log 
represented as a directed acyclic graph (happen-before graph). This 
inverse-log is later used in the validation phase, to replay the block 
creator's parallelization schedule.
Anjana et al.~\cite{anjanConcurrent} replaced the pessimistic lock with 
Optimistic Concurrency Control favoring low-conflict workloads but at the 
cost of high abort rate for transactions with more conflicts. 
Zhang et al.~\cite{zhang2018enabling} improves concurrency of the 
validation phase by recording the write set of each transaction in the 
block, which incurs additional storage and communication overhead.
These works are optimistic in the sense that the adversary  
can always create blocks whose validation cannot be parallelized.
Hence, they are not suitable block validation time in the 
presence of a Byzantine adversary. 

Blockchain sharding based solutions~\cite{luu2016secure,Omniledger,
ChainSpace,zamani2018rapidchain,wang2019monoxide} partition the 
set of nodes in the system into smaller groups, known as
{\em shards}, where each shard maintains a subset of blockchain states and processes a subset of transactions. 
They scale overall transaction processing capacity 
of the system in proportion to the number of shards in the 
system. However, the processing capacity of each shard 
is still limited, and hence they are not suitable for 
executing computationally intensive transactions within
each shard. 
Furthermore, cross-shard transactions, i.e., 
transactions that modify state from more than one shard 
still suffer from longer latency~\cite{ruan2019blockchains}.

\section{Discussion}
\label{sec:faq}
We discuss some questions about \prot\ in this section.

\one
\noindent{\bf How \prot\ handle the invalid transactions? (contract has an error or runs out of gas)}  Such invalid transactions are treated in the same way as that of Ethereum. Since the execution of transactions is deterministic, if a transaction fails at a node, it will fail at every other node. So each node will locally immediately roll back the states modified by the failed transaction.

\one
\noindent{\bf Can \prot\ be useful for other consensus protocols besides PoW? } 
\prot's core idea of delaying execution of transactions by up to
$\ths$ blocks can be applied to other consensus protocols such as Algorand, proof-of-stake, and BFT protocols. Determining the gains it provides with these protocols is a potential future research direction.

\one 
\noindent{\bf How realistic is the assumption of separate processing units (hardware) for transaction validation and proof of work (consensus)? }
\prot\ will be adopted by blockchains such as Bitcoin and
Ethereum where block processing can be done using CPUs while
mining requires ASICs. Furthermore, the number of simultaneous
forks in them are quite small~\cite{gencer2018decentralization}.



\one
\noindent{\bf How does \prot\ handle light clients? } Blocks in \prot\ maintain the cryptographic digest of the state, which allows light clients to efficiently validate or prove to
other light clients, a portion of the latest state similar to Ethereum.

\one
\noindent{\bf How \prot\ differs from off-chain solutions that do not make additional security assumptions as that of the PoW blockchain like Rollups?}
Rollups is the second layer solution that requires "operators" to stake a bond in the rollup contract. This is not proposed as the smart contract scaling solution but instead a high throughput solution that enhances the transaction throughput between the operators involved in the bond. It also fails in the case of interactive smart contract transactions.


\one
\noindent{\bf Why there is need to modify the standard longest chain rule?}
Since validation times of blocks are large, a miner may be presented with a long chain which it cannot immediately validate. By mining on the longest validated chain, we ensure that a miner does not mine on any chain which contains a block with an invalid state. In addition, honest miners will not stop mining even if a long unvalidated chain is known to it. We use this fact in our security proofs against a general adversary.

\section{Conclusion}
\label{sec:discussion}
We have presented \prot\, which theoretically allows validation
time of blocks in PoW based blockchains to be comparable to the
average interarrival time, i.e. $\tau/\intv \approx 1$.
Such a high validation time allows \prot\ to scale execution 
of smart contracts on-chain. 
Hence, it makes blockchains accessible to 
applications with heavy computation. 
Another advantage of the on-chain approach is that all miners 
update state locally and hence obviate the need for transferring
state updated due to transaction execution. Hence the bandwidth 
usage of \prot\ is identical to existing system such as Ethereum.

We prove the security of \prot\ in synchronous network with end-to-end
delay of $\Delta$ in the presence of a Byzantine adversary 
considering all possible adversarial strategies.
We also present a principled approach to pick $\ths$ for any given 
choice of parameters.

Although, state corresponding to a contract transactions gets
reported $\ths$ block later, our analysis and evaluation demonstrate
that, most~$(>99\%)$ blocks in \prot\ finds a queue size of less
than five on its arrival. Hence in practice, miners will have execution
results of transactions reasonably quickly. Furthermore, 
$\toki$ transactions are executed
immediately thus can be used in the low latency applications.
Our experimental results demonstrate working of \prot\ for 
$\tau/\intv = 0.70$ for an implementation over a standard Ethereum 
geth client.

\section*{Acknowledgment}
The authors would like to thank Andrew Miller and Manoj M. Prabhakaran for their feedback on the early version of the paper. The work is supported in part by a generous grant of cloud credits from Oracle Corp., that we have used to run all of the experiments whose results are reported in this paper.

\Urlmuskip=0mu plus 1mu\relax
\bibliographystyle{ACM-Reference-Format}
\bibliography{references}

\appendix
\section{Block Processing time in Ethereum}
\label{apx:eth block processing time}
We measure $\tau$, the time required to execute all transactions in 
Ethereum blockchain for the first 7.5 Million blocks using Ethereum
Geth Client. Our measurement using a virtual machine with 16 
cores, 120GB memory, 6.4TB NVMe SSD, and 8.2 Gbps network 
bandwidth shows that $\tau$ is only about $1\%$ of the average interarrival time. 
Figure~\ref{fig:execution time} illustrates 
our measurement for all the blocks.
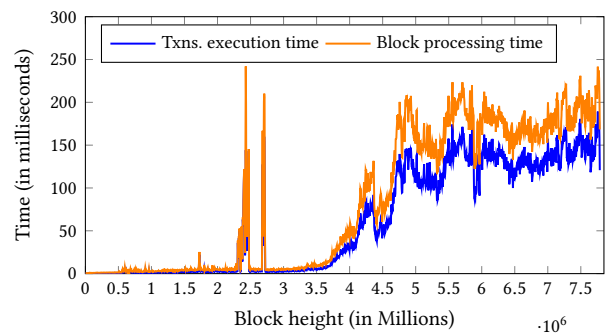
\begin{figure}[h!]
    \centering
    \pgfplotsset{footnotesize,height=5cm, width=\linewidth}
    \begin{tikzpicture}
    \begin{axis}[
        legend pos=north west,
        ymin=0.0,
        ymax=300,
        xmin=0.0,
        xmax=7850000,
        legend columns=2,
        xlabel={Block height (in Millions)},
        ylabel= {Time (in milliseconds)},
        grid=minor,
        ] 
        \addplot[line width=0.75, mark=none, blue] table [x=blockHeight, y=xTime, col sep=comma] {data/timeMeasurement.csv};
        \addplot [line width=0.75, mark=none, orange] table [x=blockHeight, y=tTime, col sep=comma] {data/timeMeasurement.csv};
        \addlegendentry{Txns. execution time}
        \addlegendentry{Block processing time}
    \end{axis}
    \end{tikzpicture}
    \caption{Time taken to validate a received block Ethereum by 
    re-executing the contained transactions. Transaction execution 
    time refers to the time taken to execute only the transactions 
    of a received block while block processing time includes both
    transaction execution time and time taken to write the updated
    state to the disk.}
    \label{fig:execution time}
\end{figure}

\section{High $\tau$ in legacy PoW Blockchains}
\label{apx:high tau ethereuem}
Recall from~\textsection\ref{sec:background}, on arrival of
new block, a miner first validates the received block for $\tau$
units of time and then creates the next block for another 
$\tau$ units of time. Only after $2\tau$ time units the miner starts
PoW for the next block. Whereas when a miner mines the block 
himself it only spends $\tau$  creating the next block
before starting PoW. Note that in case a miner receives the
next block while validating the previous one, the time
for which the miner needs to wait before it could start PoW
is higher than $2\tau$. The exact waiting time depends on the
exact time of arrival. We make a simplification and  assume 
that all miners~(including adversary \adv) 
releases blocks only after $\tau$ units of time has passed
from the broadcast of the previous block. We also assume that
the network is fully synchronous. 

Let $\Lambda = \{\lambda_j|j=1,2,\ldots,|N|\}$ be the 
block  arrival rates due to a miner $n_j$ during the periods they are performing PoW. 
Let $\tau_1$and $\tau$ where $\tau_1 = c\tau$ for 
$0\le c\le 1$ denote the time 
required to validate a full block by miner $n_1$ and 
other miners $n_j$ for $j>1$ respectively. Note that $c<1$ 
implies that $n_1$ can validate (and also create) blocks faster 
than others. For example, $c=1/2$ implies $n_1$ takes half the time
than others take to validate or create blocks.
Let $U=\{1,2,\cdots,|N|\}$ be the states of the MC where
state $u$ represents that the miner of the latest 
block is $n_u$. State transition happens on arrival of every
block and $p_{u,v}$ denote the transition probability for
the transition from state $u$ to state $v$. Note that
on every state transition to state $v$, $n_u$ mines 
a single block. 

\noindent
{\bf All honest miners.}
Since all miners in Ethereum do not always start 
PoW simultaneously, transition probabilities 
in the MC must consider only active miners.
For example, for $c>1/2$, in state $1$, no miner performs PoW 
for the first $c\tau$ time units. Between time $(c\tau,2\tau]$ 
only $n_1$ does PoW while other miners were busy validating
and creating blocks. On the contrary,
at any other state $v$, $n_v$ starts PoW at time $\tau$,
$n_1$ starts at $2c\tau$ and the remaining start at time 
$2\tau$. This is depicted in the top diagram of 
Figure~\ref{fig:eth honest timing}. Similarly, the bottom
diagram in the same Figure illustrates the miners that
perform PoW for at different intervals starting from 
different state. 
In this paper we only derive transaction probabilities 
of the MC for $c>1/2$ as one can easily derive for 
$c\le1/2$ using similar approach.
\begin{figure}[t!]
    \centering
    \includegraphics[height=3.5cm]{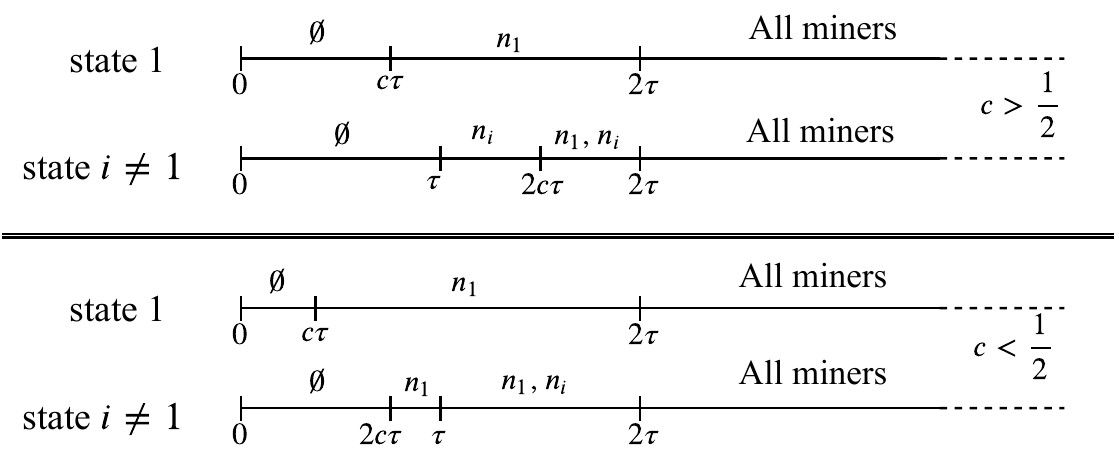}
    \caption{Miners who are active during each time interval from 
    the instant when last block gets generated in a situation 
    where all miners are  honest.}
    \label{fig:eth honest timing}
\end{figure}

\begin{lemma}
Let $X,Y$ be two independent random variable with exponential 
distribution with arrival rate $\lambda_x, \lambda_y$ respectively. 
Then the probability of the event $A = (X\le \tau \land X\le Y)$ 
is denoted using $N(\lambda_x, \lambda_y, \tau)$ and is equal to:
\begin{align}
\text{Pr}[A] & = N(\lambda_x, \lambda_y, \tau) \notag \\
             & = \frac{\lambda_x}{\lambda_x + \lambda_y}\left(1-e^{-(\lambda_x + \lambda_y)\tau}\right) 
\label{eq:compete block}
\end{align}
\label{lem:compete block}
\end{lemma}
\noindent\begin{proof}
For any arbitrary time $t$ where $0<t\le\tau$, 
Pr$[t-dt\le X \le t] = f_X(t)dt = \lambda_xe^{-\lambda_xt}dt$ 
and Pr$[Y>t] = 1 - \lambda_ye^{-\lambda_yt}$. Thus to get the 
closed form expression for Pr$[A]$, we compute,
\begin{align}
\text{Pr}[A] &= \int_{0}^{\tau}f_X(t)\text{Pr}[Y>t]dt \notag \\
             &= \int_{0}^{\tau}\lambda_xe^{-\lambda_xt}e^{-\lambda_yt}dt
\end{align}
Solving the above produces the desired result.
\end{proof}

\begin{theorem}
For a particular $c>1/2,\tau$ and 
$\Lambda=\{\lambda_j|j=1,2,\ldots,|N|\}$, with every miner 
honestly following the protocol, the state transition probabilities 
are given as:
\begin{align}
p_{1,1} &= 1-M_0(\lambda_1, 2\tau-c\tau) + p_1 M_0(\lambda_1,2\tau-c\tau) \\
p_{u,1} &= M_0(\lambda_u,2c\tau-\tau)\left[N(\lambda_1,\lambda_u, 2\tau - c\tau) \right. \notag\\ 
        &  \left. \qquad+ M_0(\lambda_1+\lambda_u, 2\tau-c\tau)p_1 \right]\\
p_{1,v} & = p_vM_0(\lambda_1, 2\tau-c\tau) \\
p_{u,u} & = 1-M_0(\lambda_1, 2c\tau-\tau) \notag \\
        & \qquad + M_0(\lambda_1, 2c\tau-\tau) \left[N(\lambda_u,\lambda_1, 2\tau - c\tau) \right. \notag \\
        & \left. \qquad + M_0(\lambda_1+\lambda_u, 2\tau-c\tau)p_u \right] \\
p_{u,v} & = M_0(\lambda_1, 2c\tau-\tau)M_0(\lambda_1+\lambda_u, 2\tau-2c\tau)p_v 
\end{align}
with $N(\lambda_u,\lambda_v,t)$ is as given in 
Lemma~\ref{lem:compete block} and $M_0(\lambda_x, t)$ is the 
probability of $0$ arrival in a Poisson process in a time 
interval $t$ with arrival rate $\lambda_x$. Hence
$M_0(\lambda_x, \tau) = e^{-\lambda_x\tau}$.
\end{theorem}
\begin{proof}
Transition of state $1$ to $1$ can happen either if $n_1$
mines the next block during $2\tau-c\tau$ interval or if 
$n_1$ mines the block after time $2\tau$. The former happens 
with a probability $1-M_0(\lambda_1, 2\tau-c\tau)$.
The latter happens with a probability $p_1$ conditioned on 
that the former event did not happen. Hence the probability of the 
latter is $p_1 M_0(\lambda_1,2\tau-c\tau)$. 
Also since these two events are mutually exclusive, $p_{1,1}$ 
is sum of the probability of the events.
Similarly, starting with state $1$, any other miner $n_v$ 
will mine the next block only if $n_1$ does not mine 
the block during an interval of length $2\tau-c\tau$  starting 
at $c\tau$. Also since all miners will be mining after 
time $2\tau$ if no block was mined before that, 
the probability that the next winner would be $n_v$ 
is $p_v$. Hence the transition probability $p_{1,v}$
equal to $p_vM_0(\lambda_1, 2\tau-c\tau)$.

Alternatively, starting from a state $u$ with $u\ne1$, $n_u$ 
will mine the next block during time interval $(\tau,2c\tau]$ 
with probability $1-M_0(\lambda_u,2c\tau-\tau)$. Otherwise
$n_u$ can mine the block during $(2c\tau,2\tau]$ . But as both 
$n_1$ and $n_u$ will be mining during $(2c\tau,2\tau]$, 
the probability of $n_u$ mining the block before $n_1$ is equal to 
$N(\lambda_u,\lambda_1,2\tau-2c\tau)$. Lastly if $n_u$ mine the
block in neither of these interval, $n_u$ will mine the next
block with probability $p_u=\lambda_u/\lambda$. 
Combining the above will give the transition probability 
$p_{u,u}$. The transition probability $p_{u,1}$ can be derived 
similarly.

Lastly, state transition from a state $u$ to $v$ with 
$u\ne v\ne1$ can happen if neither $n_1$ nor $n_u$ mines a
block prior time $2\tau$. Hence the transition probability 
$p_{u,v}$ is equal to 
$M_0(\lambda_1,2c\tau-\tau)M_0(\lambda_1+\lambda_u,2\tau-2c\tau)p_v$
\end{proof}

Using the above state transition probabilities and mining power 
from Table~\ref{tab:mining fraction} we numerically compute
the stationary probabilities of the Markov chain with all miners being honest. 
Figure~\ref{fig:high tau honest} 
illustrates our results for different $c$ with varying $\tau$.
%

\noindent
{\bf Higher $\tau$ in the presence of an Adversary.}
Let node $n_1$ with arrival rate $\lambda_1$ be controlled
by an adversary \adv. We consider two different
behaviors of the adversarial node $n_1$. First, $n_1$ 
validates the received blocks as per the protocol but instantly
creates a block by putting transactions whose execution 
results are already known to $n_1$. In second $n_1$ skips 
validation of the received block as well and instantly 
starts PoW on top a new full block. The former attack is 
very practical as any miner can do that without any additional
computational resources. The later damages fairness
more severely but requires $n_1$ to produce final state due 
to transactions in the received block without executing them. 
An adversary can launch the later attack if it can download
the modified state due to previous block from the creator of 
of the previous block. 
Figure~\ref{fig:eth skip timing}
illustrates which miners do PoW at different time intervals
starting from the instant of successful PoW on the previous 
block. The diagram at the top is when $n_1$ skips only creation
and in the diagram at the bottom is when $n_1$ skips both validation 
and creation. Here, we will only derive the transition
probabilities for the latter, and the transition probabilities
for the former can be derived similarly. 
\begin{figure}[t!]
\centering
\footnotesize
\includegraphics[height=3cm]{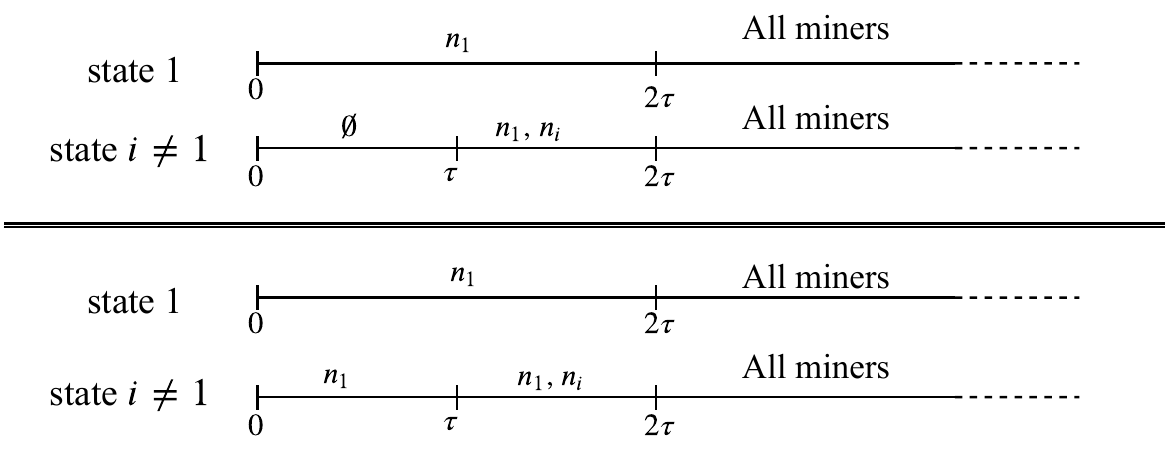}
\caption{Miners who are active during each time interval from 
the instant when last block gets generated. The diagram in the
top corresponds to a adversarial node $n_1$, who only skips creation 
of blocks and the diagram at the bottom is in 
when $n_1$ skips both validation of received block and 
creation of new ones.}
\label{fig:eth skip timing}
\end{figure}

\begin{theorem}
Given $\tau$ and $\Lambda=\{\lambda_j|j=1,2,\ldots,|N|\}$, 
with $\lambda_1$ as arrival rate of the adversarial node $n_1$,
the transition probabilities for the Markov chain when $n_1$
skips validation of received blocks and the creation of new ones
are:
\begin{align}
p_{1,1} &= 1-M_0(\lambda_1, 2\tau) + p_1M_0(\lambda_1, 2\tau) \\
p_{1,v} &= p_vM_0(\lambda_1, 2\tau)\\
p_{u,1} &= 1-M_0(\lambda_1, \tau) + M_0(\lambda_1, \tau)[N(\lambda_1,\lambda_u,\tau) \notag\\
        & \qquad+ M_0(\lambda_1+\lambda_u, \tau)p_1] \\
p_{u,u} &= M_0(\lambda_1, \tau)[N(\lambda_u,\lambda_1,\tau) + M_0(\lambda_1+\lambda_u, \tau)p_u] \\
p_{u,v} &= M_0(\lambda_1, \tau)M_0(\lambda_1+\lambda_u, \tau)p_v
\end{align}
\end{theorem}
\begin{proof}
When $n_1$ mines the last block, all other nodes will start
PoW for the next block only after $2\tau$ time units whereas
$n_1$ will do PoW for the entire $2\tau$ interval. Hence
transition from state $1$ to $1$ will happen if either
$n_1$ mines the block in the first $2\tau$ time interval 
or $n_1$ mine the after $2\tau$ time units. 
The former happens with a 
probability $1-M_0(\lambda_1, 2\tau)$ and the later can happen
with probability $p_1$ conditioned on the former not happening.
Hence $p_{1,1}$ is 
$1-M_0(\lambda_1, 2\tau) + p_1M_0(\lambda_1, 2\tau)$. Similarly,
the transition from state $1$ to another state can only happen
if $n_1$ did not mine during the first $2\tau$ time units. Also
since all nodes will mining after $2\tau$ time units the transition
probability $p_{1,v}$ is equal to $p_vM_0(\lambda_1, 2\tau)$.

When a node $n_u, u\ne1$ mines the last block, $n_1$ instantly
starts PoW for the next block. Hence for the first $\tau$ units 
of time only $n_1$ will be mining as even $n_u$ will be busy 
creating the next block. During time $(\tau,2\tau]$ both
$n_1$ and $n_u$ be mining and after $2\tau$ the rest of the
miners will start PoW for the next block. Thus $n_1$ can 
mine the next block either during the first $\tau$  or in the
interval $(\tau,2\tau]$ or after $2\tau$. The first can happen with
a probability $1-M_0(\lambda_1, 2\tau)$, the second 
with probability $N(\lambda_1,\lambda_a,\tau)$ conditioned on
that the first did not occur, and lastly $n_1$ will mine a 
block after time $2\tau$ with probability $p_1$ in case no block was mined prior to $2\tau$. Combining 
the above will give us the transition probability $p_{u,1}$. 

Similarly transition from state $u$ to itself happens 
when $n_u$ mines the next block either during time interval
$(\tau,2\tau]$ or after time $2\tau$. The former happens
with probability $N(\lambda_u,\lambda_1,\tau)$ conditioned
on the event that  $n_1$ did not mine the next block during first $\tau$
time units and the later with probability $p_u$ conditioned on
neither $n_u$ nor $n_1$ mining a block before $2\tau$.
Finally, transition to state to a state $v,v\ne u\ne1$, will 
only happen with if neither $n_u$ nor $n_1$ mine the next block
before $2\tau$. Hence the transition probability $p_{u,v}$ is 
equal to $M_0(\lambda_1, \tau)M_0(\lambda_1+\lambda_u, \tau)p_v$. 
\end{proof}
Figure~\ref{fig:high tau adv} illustrates numerically
computed fraction of blocks mined by node $n_1$ which
can validate blocks faster (by factor $1/c$) and skips creation of blocks. Note
that $c=0.0$ in the figure corresponds to the case where
$n_1$ skips both validation of received blocks and creation
of next ones. 
\begin{figure}[t!]
    \centering
    \pgfplotsset{footnotesize,height=4.5cm, width=0.95\linewidth}
    \begin{tikzpicture}
    \begin{axis}[
        legend pos=north west,
        ylabel=Fraction of blocks mined,
        xlabel=$\tau/\intv$,
        grid=minor,
        xtick={0.011, 0.122, 0.205, 0.260},
        ]
        \addplot table [x=t, y=c0, col sep=comma] {data/stat-advc.csv};
        \addplot table [x=t, y=c20, col sep=comma] {data/stat-advc.csv}; 
        \addplot [mark=x, orange] table [x=t, y=c50, col sep=comma] {data/stat-advc.csv}; 
        \addplot table [x=t, y=c60, col sep=comma] {data/stat-advc.csv}; 
        \addplot table [x=t, y=c100, col sep=comma] {data/stat-advc.csv}; 
        \addlegendentry{$c=0.00$}
        \addlegendentry{$c=0.20$}
        \addlegendentry{$c=0.50$}
        \addlegendentry{$c=0.60$}
        \addlegendentry{$c=1.00$}
    \end{axis}
    \end{tikzpicture}
    \caption{Fraction of blocks mined by a adversary who 
    validates a received block in time $\tau_1 = c\tau$ units 
    of time, instantly creates block by putting transactions 
    whose execution results are already known to the adversary,
    and controls approximately $~0.33$ fraction of the mining 
    power.}
    \label{fig:high tau adv}
\end{figure}
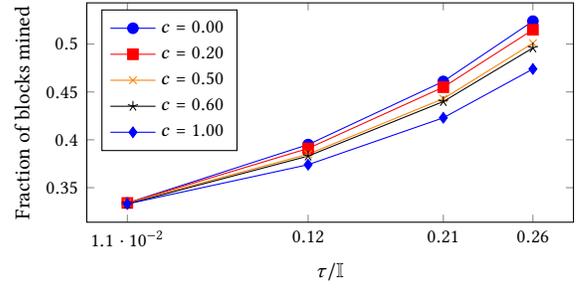

\noindent
{\bf Closed-form probabilities.}
In a system with $n$ miners, one needs to solve a system of $n$ 
linear equations to get closed-form equations for the stationary 
probabilities. We do not solve these as a part of this paper.
But we compute closed-form stationary distribution for a 
particular case is given below.

Let there be $K$ nodes in the network with an equal mining power of 
each node. Among these $K$ nodes an adversary \adv controls a $f$ 
fraction of the node and all adversarial nodes skips both 
validation and creation of blocks. 

Consider an honest node which has just mined a 
block and just finished creating a new block. We assume 
that $K$ is large enough to ensure that the probability of 
this honest miner successfully mining in the next
 $\tau$ 
units 
is approximately zero. In other words, we assume that the 
probability of all honest nodes together mining in this interval 
is zero and so only \adv\ mines in this interval. 
Under this assumption the Markov 
chain discussed so far can be reduced to a MC with only two 
states as depicted in Figure~\ref{fig: reduced mc}. 
\begin{figure}[t!]
\includegraphics[width=\linewidth]{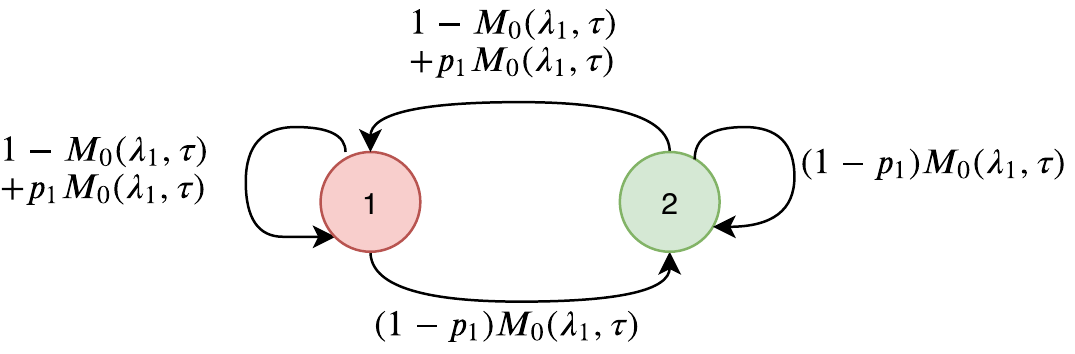}
\caption{Reduced Markov chain to analyze Skipping adversary. A 
general Markov chain for the skipping adversary can be reduced to 
the Markov chain above for large $n$ where each adversary controls 
a $f$ fraction of compute power and honest mining power is 
distributed among the remaining $n-1$ honest power such that 
the maximum amount of compute power a single miner controls extremely 
small.}
\label{fig: reduced mc}
\end{figure}

State 1 (resp. state 2) represents that the last state was mined 
by a adversarial (resp. honest) node. 
The transition
probabilities are given as:
\begin{align}
p_{0,0} &=p_{1,0} = 1-M_0(f\lambda, 2\tau) + fM_0(f\lambda, 2\tau) \\
p_{0,1} &=p_{1,1} = (1-f)M_0(f\lambda, 2\tau)
\end{align}
Let $\trans_r$ be the transition probability matrix of this Markov 
chain and let $\Psi=(\psi_1, \psi_2)$ be the stationary 
probabilities of the states $1,2$ respectively. Then we solve, 
$\Psi\trans_r = \Psi$ to get the stationary probabilities and 
these values are:
\begin{align}
\psi_1 &= 1-M_0(f\lambda,2\tau) + fM_0(f\lambda,2\tau),\\ 
\psi_2 &= 1-\psi_1
\end{align}

\section{Queue overflow in honest execution.}
In this section we will evaluate the performance of \prot\ 
when all parties are honest as we believe this will be the 
most likely case.

For any given $\ths,\lambda,\Delta$, and $\tau$ we compute the bounds 
using equation~\ref{eq:zeta bound} by putting $\beta=0$ and $t^*=0$.
This corresponds to a execution of \prot\ in the absence of an adversary.
Figure~\ref{fig:honest bound} illustrates the upper bound on the fraction
of time queue at a honest miner will have more than $\ths$ blocks for
different values of $\ths$. All plots are for $\tau/\intv=0.5$.  
Note that in the absence of \adv, for $\intv/\Delta=10$, with $\ths$ as low 
as $20$ queue, less than one in a billion honest blocks will hit a queue larger than $\ths$. 
\pgfplotsset{small,label style={font=\fontsize{8}{9}\selectfont},legend style={font=\fontsize{7}{8}\selectfont},height=4.5cm,width=0.95\linewidth}
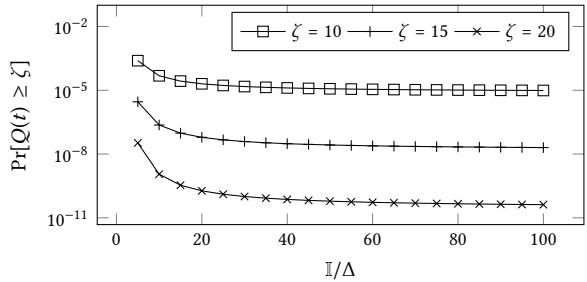
\begin{figure}[h!]
    \centering
    \begin{tikzpicture}
    \begin{semilogyaxis}[
        legend pos=north east,
        legend columns=3,
        xlabel={$\intv/\Delta$},
        ylabel= {Pr$[Q(t)\ge\ths]$},
        grid=none,
        ymax=0.1,
        ]
        \addplot[mark=square, black] table [x=c, y=10-75, col sep=comma] {data/h-Data.csv}; 
        \addplot[mark=+, black] table [x=c, y=15-75, col sep=comma] {data/h-Data.csv};
        \addplot [mark=x, black] table [x=c, y=20-75, col sep=comma] {data/h-Data.csv};
        \addlegendentry{$\ths=10$}
        \addlegendentry{$\ths=15$}
        \addlegendentry{$\ths=20$}
    \end{semilogyaxis}
    \end{tikzpicture}
    \caption{Upper bound on the probability of queue at a honest miner
    crossing the threshold in the absence of an adversary for $\tau/\intv=0.5$}
    \label{fig:honest bound}
\end{figure}

\newpage
\section{Interaction between contracts}
\label{apx:interactive contracts}
We measure the interaction between contracts for 50k blocks 
starting at a height of 6.5 Million. To measure this 
information, we sync an Ethereum Geth Client in ${\tt archive}$
mode. Such a node stores the ${\tt debug\_trace}$ of all 
transactions starting from the genesis block. We loop through
debug traces of each transaction in every block in the 
given range and use EVM ${\tt CALL, DELEGATECALL,}$ and 
${\tt STATICCALL}$ opcode to determine whether the transaction 
invoke a function which internally calls functions from 
other transactions. Figure~\ref{fig:contract interaction}
illustrates our findings. Specifically, we observe that
in this corresponding range, in each block more than 
$50\%$ of the transactions are addressed to 
a smart contract. Also, among all the transactions
approximately $20\%$ of the transactions internally
invokes function calls to other contracts. 
\begin{figure}[h!]
    \centering
    \pgfplotsset{footnotesize,height=4.5cm, width=0.95\linewidth}
    \begin{tikzpicture}
    \begin{axis}[
        legend pos=north west,
        xlabel={Block Number},
        ylabel= {Fraction },
        grid=minor,
        ymax=0.80,
        legend columns=2,
        ]
        \addplot table [x=blockHeight, y=intContRatio, col sep=comma] {data/contInt65.csv}; 
        \addplot[mark=+, blue] table [x=blockHeight, y=totalContRation, col sep=comma] {data/contInt65.csv};
        \addlegendentry{Interactive contract txns}
        \addlegendentry{Contract txns}
    \end{axis}
    \end{tikzpicture}
    \caption{Fraction of contract transactions in a block and the fraction of contract transactions which invokes at least one function from a different contract for 50 thousands blocks starting with block height 6.5 Million.}
    \label{fig:contract interaction}
\end{figure}
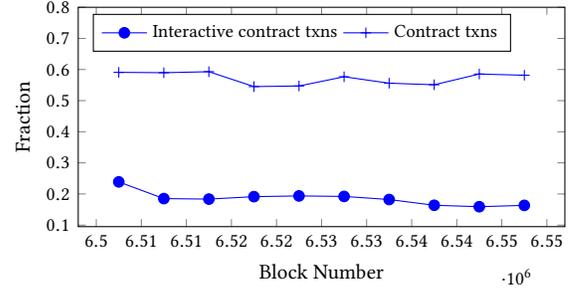

\end{document}